\documentclass[twoside]{article}

\usepackage[preprint]{aistats2022}
\usepackage[round]{natbib}

\usepackage{hyperref}
\usepackage{cuted}
\usepackage[textsize=tiny]{todonotes}

\usepackage{xspace}
\usepackage{bbm}
\input{mathlig}
\usepackage{mathtools}
\usepackage{relsize}
\usepackage{mathrsfs}
\usepackage{dsfont}
\DeclarePairedDelimiterX{\inp}[2]{\langle}{\rangle}{#1, #2}
\makeatletter
\newcommand*\bigcdot{\mathpalette\bigcdot@{.5}}
\newcommand*\bigcdot@[2]{\mathbin{\vcenter{\hbox{\scalebox{#2}{$\m@th#1\bullet$}}}}}
\makeatother

\newcommand{\muspace}{\mspace{1mu}}

\DeclareRobustCommand{\scond}{\mathchoice{\muspace\vert\muspace}{\vert}{\vert}{\vert}}
\mathlig{|}{\scond}

\DeclareRobustCommand{\discint}{\mathchoice{\mspace{-1.5mu}:\mspace{-1.5mu}}{\mspace{-1.5mu}:\mspace{-1.5mu}}{:}{:}}
\mathlig{::}{\discint}
\newcommand{\suchthat}{\mathchoice{\colon}{\colon}{:\mspace{1mu}}{:}}

\def\sgn{\mathop{\rm sgn}\nolimits}%
\newcommand{\Ac}{\mathcal{A}}

\newcommand{\Cc}{\mathcal{C}}

\newcommand{\Hc}{\mathcal{H}}

\newcommand{\Mc}{\mathcal{M}}

\newcommand{\Tc}{\mathcal{T}}

\newcommand{\av}{{\bf a}}

\newcommand{\fv}{{\bf f}}
\newcommand{\gv}{{\bf g}}

\newcommand{\xv}{{\bf x}}

\newcommand{\uv}{{\bf u}}
\newcommand{\vv}{{\bf v}}

\newcommand{\xb}{{\mathbf x}}

\newcommand{\yb}{{\mathbf y}}

\newcommand{\yh}{{\hat{y}}}

\newcommand{\xt}{{\tilde{x}}}

\def\a{\alpha}
\def\b{\beta}

\def\e{\epsilon}

\def\th{\theta}

\newcommand\eg{e.g.,\xspace}
\newcommand\ie{i.e.,\xspace}
\def\textiid{i.i.d.\@\xspace}
\newcommand\iid{\ifmmode\text{ i.i.d. } \else \textiid \fi}

\newcommand{\Real}{\mathbb{R}}

\newcommand{\half}{\frac{1}{2}}

\def\mathllap{\mathpalette\mathllapinternal}
\def\mathllapinternal#1#2{%
  \llap{$\mathsurround=0pt#1{#2}$}}

\def\clap#1{\hbox to 0pt{\hss#1\hss}}
\def\mathclap{\mathpalette\mathclapinternal}
\def\mathclapinternal#1#2{%
  \clap{$\mathsurround=0pt#1{#2}$}}

\let\oldstackrel\stackrel
\renewcommand{\stackrel}[2]{\oldstackrel{\mathclap{#1}}{#2}}

\DeclarePairedDelimiterX{\infdivx}[2]{(}{)}{%
  #1\;\delimsize\|\;#2%
}

\renewcommand{\hbar}{h\mathllap{\overline{\vphantom{h}\hphantom{\rule{4.6pt}{0pt}}}\mspace{0.77mu}}}

\catcode`~=11 
\newcommand{\urltilde}{\kern -.06em\lower -.06em\hbox{~}\kern .02em}
\catcode`~=13 

\DeclarePairedDelimiterX{\norm}[1]{\lVert}{\rVert}{#1}
\DeclarePairedDelimiterX{\abs}[1]{\lvert}{\rvert}{#1}

\usepackage{xparse}

\let\oldpartial\partial
\renewcommand*{\partial}{\mathop{}\!\oldpartial}
\newcommand{\defeq}{\mathrel{\mathop{:}}=}

\newcommand{\numberthis}{\addtocounter{equation}{1}\tag{\theequation}}

\makeatletter
\let\save@mathaccent\mathaccent
\newcommand*\if@single[3]{%
  \setbox0\hbox{${\mathaccent"0362{#1}}^H$}%
  \setbox2\hbox{${\mathaccent"0362{\kern0pt#1}}^H$}%
  \ifdim\ht0=\ht2 #3\else #2\fi
  }
\newcommand*\rel@kern[1]{\kern#1\dimexpr\macc@kerna}
\newcommand*\wideaccent[2]{\@ifnextchar^{{\wide@accent{#1}{#2}{0}}}{\wide@accent{#1}{#2}{1}}}
\newcommand*\wide@accent[3]{\if@single{#2}{\wide@accent@{#1}{#2}{#3}{1}}{\wide@accent@{#1}{#2}{#3}{2}}}
\newcommand*\wide@accent@[4]{%
  \begingroup
  \def\mathaccent##1##2{%
    \let\mathaccent\save@mathaccent
    \if#42 \let\macc@nucleus\first@char \fi
    \setbox\z@\hbox{$\macc@style{\macc@nucleus}_{}$}%
    \setbox\tw@\hbox{$\macc@style{\macc@nucleus}{}_{}$}%
    \dimen@\wd\tw@
    \advance\dimen@-\wd\z@
    \divide\dimen@ 3
    \@tempdima\wd\tw@
    \advance\@tempdima-\scriptspace
    \divide\@tempdima 10
    \advance\dimen@-\@tempdima
    \ifdim\dimen@>\z@ \dimen@0pt\fi
    \rel@kern{0.6}\kern-\dimen@
    \if#41
      #1{\rel@kern{-0.6}\kern\dimen@\macc@nucleus\rel@kern{0.4}\kern\dimen@}%
      \advance\dimen@0.4\dimexpr\macc@kerna
      \let\final@kern#3%
      \ifdim\dimen@<\z@ \let\final@kern1\fi
      \if\final@kern1 \kern-\dimen@\fi
    \else
      #1{\rel@kern{-0.6}\kern\dimen@#2}%
    \fi
  }%
  \macc@depth\@ne
  \let\math@bgroup\@empty \let\math@egroup\macc@set@skewchar
  \mathsurround\z@ \frozen@everymath{\mathgroup\macc@group\relax}%
  \macc@set@skewchar\relax
  \let\mathaccentV\macc@nested@a
  \if#41
    \macc@nested@a\relax111{#2}%
  \else
    \def\gobble@till@marker##1\endmarker{}%
    \futurelet\first@char\gobble@till@marker#2\endmarker
    \ifcat\noexpand\first@char A\else
      \def\first@char{}%
    \fi
    \macc@nested@a\relax111{\first@char}%
  \fi
  \endgroup
}
\makeatother

\usepackage{amsmath,amsthm,amssymb}
\usepackage{xcolor}
\usepackage{enumitem}
\usepackage{wrapfig}

\usepackage{algorithm,algpseudocode}
\newcommand{\Hyperparameter}{
\textbf{Parameters} }

\algnewcommand\algorithmicforeach{\textbf{for each}}
\algdef{S}[FOR]{ForEach}[1]{\algorithmicforeach\ #1\ \algorithmicdo}

\theoremstyle{plain}
\newtheorem{theorem}{Theorem}
\numberwithin{theorem}{section}
\newtheorem{lemma}[theorem]{Lemma}
\newtheorem{corollary}[theorem]{Corollary}
\theoremstyle{remark}
\newtheorem{remark}{Remark}
\numberwithin{remark}{section}
\newtheorem{example}{Example} 
\numberwithin{example}{section}

\numberwithin{definition}{section}

\allowdisplaybreaks
\numberwithin{equation}{section}
\numberwithin{theorem}{section}

\usepackage{thmtools}
\usepackage{thm-restate}

\declaretheoremstyle[
  headfont=\color{red}\normalfont\bfseries,
  bodyfont=\color{red}\normalfont\itshape,
]{colored}

\declaretheorem[
  name=Proposition,
  sibling=theorem,
]{proposition}

\declaretheoremstyle[
  headfont=\color{blue}\normalfont\bfseries,
  bodyfont=\color{blue}\normalfont\itshape,
]{blue}

\usepackage{tcolorbox} 

\newtcolorbox{mycolorbox}[3][]
{
  colframe = #2!25,
  colback  = #2!10,
  coltitle = #2!20!black,  
  title    = {\bf #3},
  #1,
}

\newif\iflong
\longtrue

\makeatletter
\newcommand*{\addFileDependency}[1]{
  \typeout{(#1)}
  \@addtofilelist{#1}
  \IfFileExists{#1}{}{\typeout{No file #1.}}
}
\makeatother

\newcommand*{\myexternaldocument}[1]{%
    \externaldocument{#1}%
    \addFileDependency{#1.tex}%
    \addFileDependency{#1.aux}%
}
\iflong
\else
    \myexternaldocument{supplement}
\fi

\newcommand{\unitball}{\mathbb{B}}

\newcommand{\ct}{\tilde{c}}
\newcommand{\qt}{\tilde{q}}

\newcommand{\ev}{\mathbf{e}}
\newcommand{\hv}{\mathbf{h}}
\newcommand{\nv}{\mathbf{n}}
\newcommand{\wv}{\mathbf{w}}
\newcommand{\Tree}{\mathbf{T}}

\newcommand{\Reg}{\mathsf{Reg}}
\newcommand{\wealth}{\mathsf{W}}
\newcommand{\Tscr}{\mathscr{T}}

\newcommand{\bkt}{b^{\mathsf{KT}}}
\newcommand{\phikt}{\phi^{\mathsf{KT}}}
\newcommand{\psikt}{\psi^{\mathsf{KT}}}
\newcommand{\Psikt}{\Psi^{\mathsf{KT}}}

\newcommand{\qtkt}{\qt^{\mathsf{KT}}}
\newcommand{\vvkt}{\vv^{\mathsf{KT}}}
\newcommand{\xkt}{w^{\mathsf{KT}}}
\newcommand{\xvkt}{\wv^{\mathsf{KT}}}

\newcommand{\psimix}{\Psi^{\mathsf{mix}}}
\newcommand{\uvmix}{\uv^{\mathsf{mix}}}
\newcommand{\vvmix}{\vv^{\mathsf{mix}}}
\newcommand{\xvmix}{\wv^{\mathsf{mix}}}

\newcommand{\Psictw}{\Psi^{\mathsf{CTW}}}
\newcommand{\uvctw}{\uv^{\mathsf{CTW}}}
\newcommand{\vvctw}{\vv^{\mathsf{CTW}}}
\newcommand{\xvctw}{\wv^{\mathsf{CTW}}}

\usepackage{graphicx}
\usepackage{calc}
\newlength{\depthofsumsign}
\newcommand{\nsum}[1][0.7]{
    \mathop{%
        \raisebox
            {-#1\depthofsumsign+1\depthofsumsign}
            {\scalebox
                {#1}
                {$\displaystyle\sum$}%
            }
    }
}
\usepackage{etoolbox}
\robustify
\robustify\addtocounter
\robustify
\robustify

\begin{document}

%
\runningtitle{Parameter-free Online Linear Optimization with Side Information via Universal Coin Betting}

%
\runningauthor{Ryu, Bhatt, Kim}
\twocolumn[
\aistatstitle{{Parameter-free Online Linear Optimization\\with Side Information via Universal Coin Betting}}

\aistatsauthor{J. Jon Ryu \And Alankrita Bhatt \And  Young-Han Kim}

\aistatsaddress{UC San Diego \And UC San Diego \And UC San Diego/Gauss Labs Inc.}
]

\begin{abstract}
A class of parameter-free online linear optimization algorithms is proposed that harnesses the structure of an adversarial sequence by adapting to some side information. These algorithms combine the reduction technique of Orabona and P{\'a}l (2016) for adapting coin betting algorithms for online linear optimization with universal compression techniques in information theory for incorporating sequential side information to coin betting.
Concrete examples are studied in which the side information has a tree structure and consists of quantized values of the previous symbols of the adversarial sequence, including fixed-order and variable-order Markov cases.
By modifying the context-tree weighting technique of Willems, Shtarkov, and Tjalkens (1995), the proposed algorithm is further refined to achieve the best performance over all adaptive algorithms with tree-structured side information of a given maximum order in a computationally efficient manner.
\end{abstract}

\section{INTRODUCTION}
In this paper, we consider the problem of online linear optimization (OLO) in a Hilbert space $V$ with norm $\|\cdot\|$.
In each round $t=1,2,\ldots$, a learner picks an action $\xv_t\in V$, receives a vector $\gv_t\in V$ with $\|\gv_t\|\le 1$, and suffers loss $\langle\gv_t,\xv_t\rangle$.
In this repeated game, the goal of the learner is to keep her \emph{cumulative regret} small with respect to any competitor $\uv$ for any adversarial sequence $\gv^T\defeq \gv_1,\ldots,\gv_T$, where the cumulative regret is defined as the difference between the cumulative losses of the learner and $\uv\in V$, \ie
\[\Reg_T(\uv)\defeq\Reg(\uv;\gv^T) \defeq\sum_{t=1}^T \langle\gv_t,\xv_t\rangle -  \sum_{t=1}^T \langle\gv_t,\uv\rangle.
\]
Albeit simple in nature, an OLO algorithm serves as a versatile building block in machine learning algorithms~\citep{Shalev-Shwartz2011}; for example, it can be used to solve online convex optimization. 

While there exist standard algorithms such as online gradient descent (OGD) that achieve optimal regret of order $\Reg_T(\uv)=O(\|\uv\|\sqrt{T})$, these algorithms typically require tuning parameters with unknowns such as the norm $\|\uv\|$ of a target competitor $\uv$. 
For example, OGD with step size $\eta=1/\sqrt{T}$ achieves $\Reg_T(\uv)=O((1+\|\uv\|^2)\sqrt{T})$ for any $\uv\in V$, while OGD with $\eta=U/\sqrt{T}$ achieves $\Reg_T(\uv)=O(U\sqrt{T})$ for any $\uv\in V$ such that $\|\uv\|\le U$; see, \eg \citep{Shalev-Shwartz2011}.
To avoid tuning parameters, several \emph{parameter-free} algorithms have been proposed in the last decade, aiming to achieve cumulative regret of order $\tilde{O}(\|\uv\|\sqrt{T})$ for any $\uv\in V$ without knowing $\|\uv\|$ a priori~\citep{Orabona2013,McMahan--Abernethy2013,Orabona2014,McMahan--Orabona2014,Orabona--Pal2016}, where $\tilde{O}(\cdot)$ hides any polylogarithmic factor in the big O notation; the extra polylogarithimic factor is known to be necessary~\citep{Orabona2013,McMahan--Abernethy2013}.

While these optimality guarantees on regret seem sufficient, they may not be satisfactory in bounding the incurred loss of the algorithm, due to the limited power of the class of static competitors $\uv$ as a benchmark.
For example, consider the adversarial sequence $\gv,-\gv,\gv,-\gv,\ldots$ for a fixed vector $\gv\in\mathbb{B}\defeq\{\xv\in V\suchthat \|\xv\|\le 1\}$. 
Despite the apparent structure (or predictability) in the sequence, the best achievable reward of any static competitor $\uv\in V$ is zero for any even $T$.
In general, the cumulative loss of a static competitor $\uv$ is
$\sum_{t=1}^T \langle \gv_t,\uv\rangle
=\langle \sum_{t=1}^T \gv_t, \uv\rangle$, and can be large if and only if the norm $\|\sum_{t=1}^T \gv_t\|$ is large, or equivalently, when $\gv_1,\ldots,\gv_T$ are well \emph{aligned}.
It is not only a theoretical issue, since, for example, when we consider a practical scenario such as weather forecasting, the sequence $(\gv_t)$ may have such a \emph{temporal structure} that can be exploited in optimization, rather than being completely adversarial.

One remedy for this issue is to consider a larger class of competitors, which may \emph{adapt} to the history $\gv^{t-1}\defeq \gv_1,\ldots,\gv_{t-1}$.
Hereafter, we use $x_t^s$ to denote the sequence $x_t,\ldots,x_s$ for $t\le s$ and $x^t\defeq x_1^t$ by convention.
For instance, in the previous example, consider a competitor which can play two different actions $\uv_{+1}$ and $\uv_{-1}$ based on the quantization $Q(\gv_{t-1}) = \sgn(\langle \fv,\gv_{t-1}\rangle)$ for some fixed $\fv\in V$; for example, we chose standard vectors $\ev_i$ for a Euclidean space $V$ in our experiments; see Section~\ref{sec:exps}.
Then the best loss achieved by the competitor class on this sequence becomes $-(T/2)\|\gv\| (\|\uv_{+1}\|+\|\uv_{-1}\|)$, which could be much smaller than 0.
We remark that, from the view of binary prediction, this example can be thought of a first-order Markov prediction, which takes only the previous time step into consideration.
Hence, it is natural to consider a $k$-th order extension of the previous example, \ie a competitor that adapts to the length-$k$ sequence $Q(\gv_{t-k}^{t-1})\defeq Q(\gv_{t-k})\dotsc Q(\gv_{t-1}) \in \{1,\bar{1}\}^k$, where we define $\bar{1}\defeq -1$.

\begin{wrapfigure}{r}{0.125\textwidth}
\vspace{-1em}
  \begin{center}
    \includegraphics[width=.125\textwidth]{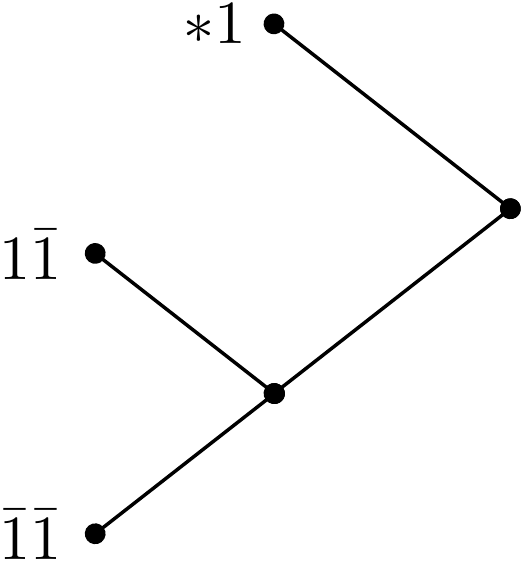}
  \end{center}
  \caption{$\Tree=\{*1,1\bar{1},\bar{1}\bar{1}\}$.}
  \label{fig:suffix_tree}
\end{wrapfigure}
We can even further sophisticate a competitor's dependence structure by allowing it to adapt to a \emph{tree structure} (also known as a \emph{variable-order Markov structure}) of the quantization sequence, which is widely deployed structure in sequence prediction; see, \eg \citep{Begleiter--El-Yaniv--Yona2004}.
For example, for the depth-2 quantization sequence $Q(\gv_{t-2}^{t-1})$, rather than adapting to the all four possible states, a competitor may adapt to the suffix falls into a set of suffixes $\Tree=\{*1,1\bar{1},\bar{1}\bar{1}\}$ of one fewer states; here, $*$ denotes that any symbol from $\{1,\bar{1}\}$ is possible in that position.
As depicted in Figure~\ref{fig:suffix_tree} for $\Tree$, in general, a suffix set has a one-to-one correspondence between a full binary tree, and is thus often identified as a tree; see Section~\ref{sec:tree_side_info} for the formal definition and further justification of the tree side information. 

Since we do not know a priori which tree structure is best to adapt to,
we ultimately aim to design an OLO algorithm that achieves the performance of the best tree competitor of given maximum depth $D\ge 1$. Since there are $O(2^{2^D})$ possible trees of depth at most $D$, it becomes challenging even for a moderate size of $D$.
We remark that the problem of following the best tree structure in hindsight, the \emph{tree problem} in short, is a classical problem which has been studied in multiple areas such as information theory~\citep{Willems--Shtarkov--Tjalkens1995} and online learning~\citep{Freund--Shapire--Singer--Warmuth1997}, but an application of this framework to the OLO problem has not been considered in the literature.




To address this problem, we combine two technical components from online learning and information theory. 
Namely, we apply an information theoretic technique of following the best tree structure for universal compression, called the \emph{context tree weighting} (CTW) algorithm invented by \citet{Willems--Shtarkov--Tjalkens1995}, to generalize a parameter-free OLO algorithm called the \emph{KT OLO algorithm} proposed by \citet{Orabona--Pal2016}, which is designed based on universal coin betting. 
Consequently, as the main result, we propose the \emph{CTW OLO algorithm} that efficiently solves the problem with only $O(D)$ updates per round achieving nearly minimax optimal regret; see Section~\ref{sec:olo_tree_side_info}.

We motivate the proposed approach by solving two intermediate, abstract OLO problems, the one with (single) side information (Section~\ref{sec:olo_single_side_information}) and the other with multiple side information (Section~\ref{sec:olo_multiple_side_information}), and propose information theoretic OLO algorithms (\ie product KT and mixture KT) respectively, which might be of independent interest. We remark, however, that it is not hard to convert any parameter-free algorithm to solve the abstract problems with same guarantees and complexity of the proposed solutions, using existing meta techniques such as a black-box aggregation scheme by \citet{Cutkosky2019} with per-state extension of a base OLO algorithm; hence, the contribution of the intermediate solutions is rather purely of intellectual merit. 
In Section~\ref{sec:exps}, we experimentally demonstrate the power of the CTW OLO algorithm with real-world temporal datasets. 
We conclude with some remarks in Section~\ref{sec:conclusion}. 
All proofs and discussion with related work are deferred to Appendix due to the space constraint.

\paragraph{Notation}
Given a tuple $\av=(a_1,\ldots,a_m)$, we use $\nsum \av\defeq \sum_{i=1}^m a_i$ to denote the sum of all entries in a tuple $\av$.
For example, we write $\nsum g^{t-1}$ to denote the sum of $g_1,\ldots,g_{t-1}$ by identifying $g^{t-1}$ as a tuple $(g_1,\ldots,g_{t-1})$.
For the empty tuple $()$, we define $\nsum ()\defeq 0$ by convention.
We use $|\av|$ to denote the number of entries of a tuple $\av$.
For a tuple of vectors $\uv_{1:S}\defeq (\uv_1,\ldots,\uv_S)\in V\times\cdots\times V$, we use $\|\uv\|_{1:S}\defeq (\|\uv_1\|,\ldots,\|\uv_S\|)\in\Real_{\ge 0}^S$ to denote the tuple of norms of each entry.



\section{PRELIMINARIES}
\label{sec:prelim}
We review the coin betting based OLO algorithm of \citet{Orabona--Pal2016}.
From this point, we will describe all algorithms in the reward maximization framework, which is philosophically consistent with the goal of gambling, to avoid any confusion, but we will keep using the conventional naming OGD even though it is actually gradient \emph{ascent}.\footnote{Note that one can translate a reward maximization algorithm to an equivalent loss minimization algorithm by feeding $-\gv_t$ instead of $\gv_t$, and vice versa.}

\subsection{Continuous Coin Betting and 1D OLO}
\label{sec:coin_1d}
Consider the following repeated gambling.
Starting with an initial wealth $\wealth_0$, at each round $t$, a player picks a \emph{signed relative bet} $b_t\in[-1,1]$.
At the end of the round, a real number $g_t\in[-1,1]$ is revealed as an outcome of the ``continuous coin toss'' and the player gains the reward  $g_t b_t \wealth_{t-1}$.
This game leads to the cumulative wealth \[\wealth_t(g^t)=\wealth_0\prod_{i=1}^t (1+g_i b_i).\]
When $g_t\in \{\pm1\}$, this game boils down to the standard coin betting, where the player splits her wealth into
$\frac{1+b_t}{2}\wealth_{t-1}$ and $\frac{1-b_t}{2}\wealth_{t-1}$, and bets the amounts on the binary outcomes $+1$ and $-1$, respectively.
It is well known that the standard coin betting game is equivalent to the binary compression, or binary log-loss prediction, which have been extensively studied in information theory; see, \eg \citep[Chapter 6]{Cover--Thomas2006}.

Even when the outcomes $g_t$ are allowed to take continuous values, many interesting connections remain to hold. 
For example, the \citet{Krichevsky--Trofimov1981}'s (KT) probability assignment, which is competitive against \iid Bernoulli models, can be translated into a betting strategy \[\bkt(g^{t-1})\defeq\bkt_t(\nsum g^{t-1}),\] where $\bkt_t(x)\defeq \frac{x}{t}$ for $x\in[-t+1,t-1]$.
As a natural continuous extension of the KT probability assignment, we define the \emph{KT coin betting potential} \[\psikt(g^t)\defeq \psikt_t(\nsum g^t)\defeq 2^t\qtkt_t(\nsum g^t),\] where
\[\qtkt_t(x)\defeq B\Bigl(\frac{t+x+1}{2},\frac{t-x+1}{2}\Bigr)\Big/B\Bigl(\half,\half\Bigr)\] for $x\in[-t,t]$ and $B(x,y)\defeq \Gamma(x)\Gamma(y)/\Gamma(x+y)$ and $\Gamma(x)$ denote the Beta function and Gamma function, respectively.
We remark that the interpolation for continuous values is naturally defined via the Gamma functions.
This simple KT betting scheme guarantees that the cumulative wealth satisfies
\[
\wealth_T(g^T) 
\ge \wealth_0 \psikt(g^T) = \wealth_0 2^T\qtkt_T(\nsum g^T)
\numberthis\label{eq:1d_coin_betting_kt_lower_bound}
\]
for any $T\ge 1$ and $g_1,\ldots,g_T\in[-1,1]$; see the proof of Theorem~\ref{thm:continuous_coin_betting_to_1D_OLO} in Appendix.
It can be easily shown that the wealth lower bound is near-optimal when compared to the best static bettor $b_t=b$ for some fixed $b\in[-1,1]$ in hindsight, the so-called Kelly betting~\citep{Kelly1956}. This follows as a simple consequence of the fact that the KT probability assignment is a near-optimal probability assignment for universal compression of \iid sequences. 
In this paper, going forward the interpretation of the coin betting potential as probability assignment in the parlance of compression will prove useful. 

In their insightful work, \citet{Orabona--Pal2016} demonstrated that the universal continuous coin betting algorithm can be directly translated to an OLO algorithm with a parameter-free guarantee.
By defining an \emph{absolute betting} $w_t\defeq b_t\wealth_{t-1}$, we can write the cumulative wealth in an additive form
\[\wealth_t(g^t)=\wealth_0+\sum_{i=1}^t g_t w_t,\]
whence we interpret $\sum_{i=1}^t g_i w_i$ as the cumulative reward in the 1D OLO with $g_1,\ldots,g_t\in[-1,1]$.
Now, if we define the KT coin betting OLO algorithm by the action \[\xkt_t\defeq \xkt(g^{t-1})=\bkt(g^{t-1})\wealth_{t-1}(g^{t-1}),\]
then the ``universal'' wealth lower bound~\eqref{eq:1d_coin_betting_kt_lower_bound} with respect to any $g^T$ can be translated to establish a ``parameter-free'' bound on the 1D regret
\[
\Reg(u;g^T)\defeq \sum_{t=1}^T g_tu - \sum_{t=1}^T g_t\xkt_t,
\]
against static competitors $u\in\Real$.
Let $(\psikt_T)^\star\suchthat\Real\to\Real$ denote the Fenchel dual of the potential function $\psikt_T\suchthat\Real\to\Real$, \ie
\[
(\psikt_T)^\star(u)\defeq \sup_{g\in\Real} (gu - \psikt_T(g)).
\]
\begin{theorem}
\label{thm:continuous_coin_betting_to_1D_OLO}
For any $g_1,\ldots,g_T\in[-1,1]$, 
the 1D OLO algorithm $\xkt_t=\bkt(g^{t-1})\wealth_{t-1}$ satisfies
\[
\sup_{u\in\Real} \Bigl\{\Reg(u;g^T) -\wealth_0(\psikt_T)^\star\Bigl(\frac{u}{\wealth_0}\Bigr)\Bigr\} \le \wealth_0.
\]
In particular, for any $u\in \Real$, we have \[
\Reg(u;g^T)
\le \sqrt{Tu^2
\ln(Tu^2/(e\sqrt{\pi}\wealth_0^2) +1)}
+ \wealth_0.\]
\end{theorem}


\subsection{Reduction of OLO over a Hilbert Space to Continuous Coin Betting}
\label{sec:coin_hilbert}
This reduction can be extended for OLO over a Hilbert space $V$ with norm $\|\cdot\|$, where we wish to maximize the cumulative reward $\sum_{t=1}^T \langle \gv_t,\xv_t \rangle$ for $\gv_1,\ldots,\gv_T\in \unitball\defeq \{\xv\in V\suchthat \|\xv\|\le 1\}$.
\citet{Orabona--Pal2016} proposed the following OLO algorithm over Hilbert space based on the continuous coin betting.
For an initial wealth $\wealth_0>0$, we define the \emph{cumulative wealth} 
\[\wealth_T(\gv^T)\defeq \wealth_0 + \sum_{t=1}^T \langle \gv_t,\xv_t \rangle\] as the cumulative reward plus the initial wealth, analogously to the coin betting.
If we define the \emph{vectorial betting} given $\gv^{t-1}$ as
\[
\vvkt(\gv^{t-1})
\defeq \bkt_t(\|\nsum\gv^{t-1}\|)\frac{\nsum\gv^{t-1}}{\|\nsum\gv^{t-1}\|}
=\frac{1}{t}\nsum\gv^{t-1}
\]
and define a \emph{potential} function
\[\Psikt(\gv^t)\defeq \psikt_t(\|\sum\gv^t\|)= 2^t\qtkt_t(\|\nsum\gv^t\|),\]
then the corresponding OLO algorithm ensures the wealth lower bound $\wealth_t(\gv^t)\ge \wealth_0\Psikt(\gv^t)$, and thus the corresponding regret upper bound in the same spirit of Theorem~\ref{thm:continuous_coin_betting_to_1D_OLO}.

\begin{theorem}[{\citealp[Theorem~3]{Orabona--Pal2016}}]
\label{thm:continuous_coin_betting_to_Hilbert_OLO}
For any $\gv_1,\ldots,\gv_T\in \unitball$, the OLO algorithm $\xvkt_t=\vvkt(\gv^{t-1})\wealth_{t-1}$ based on the coin betting satisfies
$\wealth_T\ge \wealth_0 \Psikt(\gv^T)$,
and moreover
\[
\sup_{\uv\in V} \Bigl\{\Reg(\uv;\gv^T)-\wealth_0(\psikt_T)^\star\Bigl(\frac{\|\uv\|}{\wealth_0}\Bigr)\Bigr\}\le \wealth_0.
\]
In particular, for any $\uv\in V$, we have
\[
\Reg(\uv;\gv^T)
\le \sqrt{T\|\uv\|^2 
\ln(T\|\uv\|^2/(e\sqrt{\pi}\wealth_0^2) +1)}
+ \wealth_0.\]
\end{theorem}



\section{MAIN RESULTS}
\label{sec:olo_side_information}

In what follows, we will illustrate how to incorporate (multiple) sequential side information based on coin betting algorithms in OLO over Hilbert space with an analogous guarantee by extending the aforementioned algorithmic reduction and guarantee translation.
In doing so, we will leverage the connection between coin betting and compression, and adopt universal compression techniques beyond the KT strategy, namely per-state adaptation (Section~\ref{sec:olo_single_side_information}), 
mixture (Section~\ref{sec:olo_multiple_side_information}), and context tree weighting techniques (Section~\ref{sec:tree_side_info}).
For each case, we will first define a potential function and introduce a corresponding vectorial betting which guarantees the cumulative wealth to be at least the desired potential function.

\subsection{OLO with Single Side Information via Product Potential}
\label{sec:olo_single_side_information}
We consider the scenario when a (discrete) side information $H=(h_t\in[S])_{t\ge 1}$ is sequentially available for some $S\ge 1$.
That is, at each round $t$, the side information $h_t$ is revealed before the plays.
As motivated in the introduction, the canonical example is a \emph{causal} side information based on the history $\gv^{t-1}$ such as a quantization of $\gv_{t-D}^{t-1}$ for some $D\ge 1$. 
Yet another example is side information given by an oracle with foresight such as $h_t=\sgn(\langle\gv_t,\fv\rangle)$, \ie the sign of the correlation between a fixed vector $\fv\in V$ and the incoming symbol $\gv_t$, as a rough hint to the future.

We define an \emph{adaptive competitor with respect to the side information $H$}, denoted as $\uv_{1:S}[H]$ for an $S$-tuple $\uv_{1:S}\defeq(\uv_1,\ldots,\uv_S)\in V\times \cdots \times V$, to play $\uv_{h_t}$ at time $t$, and let $\Cc[H]\defeq \{\uv_{1:S}[H]\suchthat \uv_{1:S}\in V\times \cdots\times V\}$ denote the collection of all such adaptive competitors.

We first observe that the cumulative loss incurred by an adaptive competitor $\uv_{1:S}[H]\in \Cc[H]$ can be decomposed with respect to the \emph{states} defined by the side information symbols, \ie 
\[
\sum_{t=1}^T \langle \gv_t,\uv_{h_t}\rangle=\sum_{s=1}^S \Bigl\langle \sum_{t\in[T]\suchthat h_t=s} \gv_t,\uv_s\Bigr\rangle.\]
Hence, a naive solution is to run independent OGD algorithms for each subsequence $\gv^t(s;h^t)\defeq (\gv_i\suchthat h_i=s,i\in[t])$ sharing the same side information $s\in[S]$; it is straightforward to show that the per-state OGD with optimal learning rates achieves the regret of order $O(\sum_{s=1}^S\|\uv_s\|\sqrt{T_s})$ with knowing the competitor norms $\|\uv\|_{1:S}$.
Like the per-state OGD algorithm, we can also extend other parameter-free algorithms such as DFEG~\citep{Orabona2013} and AdaNormal \citep{McMahan--Orabona2014} to adapt to side information; see Appendix~\ref{sec:other_per_state_algos}. 
This is what we call the \emph{per-state extension} of an OLO algorithm.

Here, we propose a different type of parameter-free per-state algorithm based on coin betting.
To compete against any adaptive competitor from $\Cc[H]$, we define a \emph{product KT potential function}
\begin{align*}
\Psikt(\gv^t;h^t)
&\defeq \prod_{s\in[S]}\Psikt(\gv^t(s;h^t))\\
&=\prod_{s\in[S]}\psikt_{t_s}(\|\nsum\gv^t(s;h^t)\|),
\end{align*}
where $t_s\defeq |\gv^t(s;h^t)|$ 
for each $s\in[S]$. 
Note that $\Psikt(\gv^t;h^t)$ is a function of the summations of the subsequences $(\sum \gv^t(1;h^t),\ldots,\sum \gv^t(S;h^t))$.
For each time $t$, we then define the vectorial KT betting with side information $h^t$ as the application of the vectorial KT betting onto the subsequence corresponding to the current side information symbol $h_t$, \ie
\[
\vvkt(\gv^{t-1};h^t)\defeq \vvkt(\gv^{t-1}(h_t;h^{t-1})).
\]

Unlike the other per-state extensions which play independent actions for each state thus allowing straightforward analyses, the per-state KT actions
\[
\xvkt_t(\gv^{t-1};h^t)= \vvkt(\gv^{t-1};h^t)\wealth_{t-1}
\numberthis\label{eq:kt_olo}
\]
depend on all previous history $\gv^{t-1}$ due to the wealth factor $\wealth_{t-1}$. 
We can establish the following guarantee with the same line of argument in the proof of Theorem~\ref{thm:continuous_coin_betting_to_1D_OLO}, by analyzing the Fenchel dual of $\Psikt(\gv^t;h^t)$.
Recall that for a multivariate function $\Psi\suchthat\Real^d\to\Real$, its Fenchel dual $\Psi^\star\suchthat\Real^d\to\Real$ is defined as
\[
\Psi^\star(\yb) \defeq \sup_{\xb\in\Real^d} (\yb^T\xb - \Psi(\xb)).
\]
\begin{theorem}
\label{thm:olo_side_information}
For any side information $H=(h_t\in[S])_{t\ge 1}$ and any $\gv_1,\ldots,\gv_T\in \unitball$, let $\phikt_{T_{1:S}}\suchthat\Real^S\to \Real$ be the Fenchel dual of the function
\[
(f_1,\ldots,f_S)\mapsto \prod_{s\in[S]}\psikt_{T_s}(f_s),
\]
where $T_s\defeq |\{t\in[T]\suchthat h_t=s\}|$.
Then, the OLO algorithm $\xvkt_t(\gv^{t-1};h^t)\defeq \vvkt(\gv^{t-1};h^t)\wealth_{t-1}$ satisfies
$\wealth_T\ge \wealth_0 \Psikt(\gv^T;h^T)$,
and moreover
\[
\sup_{\uv_{1:S}} \Bigl\{\Reg(\uv_{1:S}[H];\gv^T)-\wealth_0\phikt_{T_{1:S}}\Bigl(\frac{\|\uv\|_{1:S}}{\wealth_0}\Bigr)\Bigr\}\le \wealth_0.
\]
In particular, 
for any $\uv_{1:S}[H]\in \Cc[H]$,
\begin{align*}
&\Reg(\uv_{1:S}[H];\gv^T)
\numberthis\label{eq:regret_product_potential}
= \wealth_0 +\tilde{O}\Biggl(\sqrt{\sum_{s = 1}^{S} T_{s}\|\uv_{s}\|^2} \Biggr).
\end{align*}
\end{theorem}

\begin{example}
Recall the ``easy'' adversarial sequence $\gv^T = (\gv, -\gv, \gv,\dotsc,-\gv)$ for some $\gv\in \unitball$ previously considered in the introduction. 
For a side information $h_t=\sgn(\langle \gv_t, \fv \rangle)$ with some $\fv \in V$,  Theorem~\ref{thm:olo_side_information} states that $\Reg((\uv_+,\uv_-);\gv^T) = \tilde{O}((\|\uv_+\|+\|\uv_-\|)\sqrt{T})$, matching the regret guarantee of the optimally tuned per-state OGD up to logarithmic factors.
Overall, the regret guarantee against adaptive competitors for the per-state KT method implies a much larger overall reward than was achieved by an algorithm competing against static competitors.
\end{example}

\begin{remark}[Cost of noninformative side information]
Consider a scenario where competitors of the form $\uv_{1:S} =(\uv,\ldots,\uv)$ with some vector $\uv\in V$ perform best; in this case, an algorithm without adapting to side information may suffice for optimal regret guarantees. 
Even in such cases with \emph{noninformative} side information, the dominant factor in the regret remains the same as the regret guarantee with respect to the static competitor class, since $\sum_{s=1}^S T_s\|\uv_s\|^2=T\|\uv\|^2$.
\end{remark}

\begin{remark}[Effect of large $S$]
\label{rem:effect_large_state}
While side information with larger $S$ may provide more levels of granularity, too large $S$ may degrade the performance of the per-state algorithms.
Intuitively, if $S\gg 1$, it is likely that we will see each state only few times, which results in poor convergence for almost every state. These are also captured in the regret guarantee; we note that the hidden logarithmic factor of the regret bound~\eqref{eq:regret_product_potential}
might incur a multiplicative factor of at most $O(\sqrt{S})$.
Similarly, in the optimal regret attained by the per-state OGD, we have $O(\sum_{s=1}^S\|\uv_s\|\sqrt{T_s})\le O(\max_{s\in[S]}\|\uv_s\|\sqrt{ST})$.
\end{remark}

\subsection{OLO with Multiple Side Information via Mixture of Product Potentials}
\label{sec:olo_multiple_side_information}

Now suppose that multiple side information sequences $\{H^{(m)}=(h_t^{(m)}\in S^{(m)})_{t\ge 1}\suchthat m\in[M]\}$ are sequentially available;
for example, each $H^{(m)}$ can be either constructed based on a different quantizer $Q_m\suchthat V\to\{1,\bar{1}\}$ and/or based on the history $\gv_{t-D_m}^{t-1}$ of different lengths $D_m\ge 0$, each of which aims to capture a different structure of $(\gv_t)$. 
In this setting, we aim to minimize the \emph{worst} regret among all possible side information, \ie
\begin{align}
&\max_{m \in [M]} \Reg (\uv_{1:S^{(m)}}[H_m];\gv^T) \nonumber\\
&\quad= \sum_{t=1}^T \langle \gv_t, \wv_t \rangle - \min_{m \in [M]} \sum_{t=1}^T \langle \gv_t, \uv^{(H)}_{h_{mt}}\rangle,
\label{eq:regretDefMultipleSideInfo}
\end{align}
which is equivalent to aiming to follow the best side information in hindsight.

We first remark that \citet{Cutkosky2019} recently proposed a simple black-box meta algorithm that combines multiple OLO algorithms achieving the best regret guarantee, which can also be applied to solving this multiple side information problem.
For example, for algorithms $(\Ac_m)_{m\in[M]}$ each of which play an action $\wv_t^{(m)}$, the meta algorithm $\Ac$ which we refer to the \emph{addition} plays $\wv_t=\sum_{m=1}^M\wv_t^{(m)}$ and guarantees the regret \[\Reg_T^{\Ac}(\uv)\le \varepsilon +\min_{m\in[M]} \Reg_T^{\Ac_m}(\uv),\] 
provided that $\Ac_m$'s suffer at most constant regret $\varepsilon$ against $\uv=0$; the same guarantee also hold for adaptive competitors. 

Rather, we propose the following information theoretic solution.
For each side information sequence $H^{(m)}$, we can apply the per-state KT algorithm from the previous section, which guarantees the wealth lower bound $\wealth_0\Psikt(\gv^t;(h^{(m)})^t)$. 
To achieve the best among the per-state KT algorithms, we consider the \emph{mixture potential}
\[
\psimix(\gv^t;\hv^{t})
=\sum_{m=1}^M w_m \Psikt(\gv^t;(h^{(m)})^{t})
\]
for some $w_1,\ldots,w_M>0$ such that $\sum_{m=1}^M w_m=1$.
Here, $\hv_t\defeq (h_t^{(1)},\ldots,h_t^{(M)})$ denotes the side information vector revealed at time $t$. 
When there exists no prior belief on how useful each side information is, one can choose the uniform weight $w_1=\ldots=w_M=1/M$ by default. 
Now, define the \emph{vectorial mixture betting} given $\gv^{t-1}$ and $\hv^t$ as
\begin{align*}
&\vvmix(\gv^{t-1};\hv^t)\defeq  \frac{\uvmix(\gv^{t-1};\hv^{t})}{\psimix(\gv^{t-1};\hv^{t-1})},\quad\textnormal{where}\\
&\uvmix(\gv^{t-1};\hv^{t})\\
&\quad\defeq \sum_{m=1}^M w_m \Psikt(\gv^{t-1};(h^{(m)})^{t-1})\vvkt(\gv^{t-1};(h^{(m)})^t),
\end{align*}
and finally define the \emph{mixture OLO} algorithm by the action
\[\xvmix_t(\gv^{t-1};\hv^t)\defeq \vvmix(\gv^{t-1};\hv^t)\wealth_{t-1}.
\numberthis\label{eq:mixture_olo}\]
In the language of gambling, the mixture strategy bets by distributing her wealth based on the weights $w_m$'s to strategies, each of which is tailored to a side information sequence, and thus can guarantee at least $w_m$ times the cumulative wealth attained by the $m$-th strategy following $H^{(m)}$ for any $m\in [M]$.

\begin{theorem}
\label{thm:olo_multiple_side_information}
For any side information $H^{(1)},\ldots,H^{(M)}$ and any $\gv_1,\ldots,\gv_T\in \unitball$, the mixture OLO algorithm~\eqref{eq:mixture_olo} satisfies
$\wealth_T\ge \wealth_0 \psimix(\gv^T;\hv^T)$,
and moreover for any $m\in[M]$, we have
\begin{align*}
&\sup_{\uv_{1:S^{(m)}}} \Bigl\{\Reg(\uv_{1:S^{(m)}}[H^{(m)}]);\gv^T)
\\&\qquad\qquad
-w_m\wealth_0 \phikt_{T_{1:S^{(m)}}}\Bigl(\frac{\|\uv\|_{1:S^{(m)}}}{w_m\wealth_0}\Bigr)\Bigr\} 
\le w_m\wealth_0.
\end{align*}
In other words, for any $m$ and any $\uv_{1:S^{(m)}}$, we have
\begin{align*}
&\Reg(\uv_{1:S^{(m)}}[H_m];\gv^T)\\
&\quad= w_m\wealth_0  + \tilde{O}\Biggl(\sqrt{\Bigl(\ln \frac{1}{w_m}\Bigr)\sum_{s_m = 1}^{S_m} T_{s_m}^{(H_m)}\|\uv_{s_m}^{(H_m)}\|^2} \Biggr).
\end{align*}
\end{theorem}

\begin{remark}[Cost of mixture]
A mixture strategy adapts to any available side information with the cost of replacing $\wealth_0$ with $w_m\wealth_0$ in the regret guarantee for each $m\in[M]$. 
Since the dependence of regret on $\wealth_0$ scales as $O(\sqrt{\ln(1 + 1/\wealth_0)}+\wealth_0)$ from Theorem~\ref{thm:olo_side_information}, 
a small $w_m$ may degrade the quality of the regret guarantee by only a small multiplicative factor $O(\sqrt{\ln(1/w_m)})$.
\end{remark}
\begin{remark}[Comparison to the addition technique]
While the mixture algorithm attains a similar guarantee to the addition technique~\citep{Cutkosky2019}, it is only applicable to coin betting based algorithms and requires a rather sophisticated aggregation step.
Thus, if there are only moderate number of side information sequences, the addition of per-state parameter-free algorithms suffices.
The merit of mixture will become clear in the next section in the tree side information problem of combining $O(2^{2^{D}})$ many components for a depth parameter $D\ge 1$, while a naive application of the addition technique to the tree problem is not feasible due to the number of side information; see Section~\ref{sec:conclusion} for an alternative solution with the addition technique. 
\end{remark}

\subsection{OLO with Tree Side Information}
\label{sec:olo_tree_side_info}
In this section, we formally define and study a tree-structured side information $H$, which was illustrated in the introduction.
We suppose that there exists an auxiliary binary sequence $\Omega=(\omega_t\in\{\pm 1\})_{t\ge 1}$, which is revealed one-by-one at the \emph{end} of each round; hence, a learner has access to $\omega^{t-1}$ when deciding an action at round $t$.
In the motivating problem in the introduction, such an auxiliary sequence was constructed as $\omega_t\defeq Q(\gv_{t})$ with a fixed binary quantizer $Q\suchthat V\to\{\pm1\}$.


\subsubsection{Markov Side Information}\label{sec:markov_side_info}

Given $\Omega=(\omega_t)_{t\ge 1}$, the most natural form of side information is the \emph{depth-$D$ Markov side information} $h_t\defeq \omega_{t-D}^{t-1}\in \{\pm1\}^D$, \ie the last $D$ bits of $(\omega_t)_{t\ge 1}$---note that it can be mapped into a perfect binary tree of depth $D$ with $2^D$ possible states.

\begin{example}
As an illustrative application of the mixture algorithm and a precursor to the tree side information problem, suppose that we wish to compete with any Markov side information of depth $\le D$.
Then, there are $D+1$ different side information, one for each depth $d=0,\ldots,D$; for simplicity, assume uniform weights $w_d=1/(D+1)$ for each depth $d$.
Then, Theorem~\ref{thm:olo_multiple_side_information} guarantees that the mixture OLO algorithm~\eqref{eq:mixture_olo} satisfies, for any depth $d=0,\ldots,D$, 
\begin{align*}
\Reg(\uv_{1:2^d}^{(d)};\gv^T)
&= \frac{\wealth_0}{D+1}  + \tilde{O}\Biggl(\sqrt{\ln (D+1)\sum_{s = 1}^{2^d} T_{s}^{(d)}\|\uv_{s}^{(d)}\|^2} \Biggr)
\end{align*}
for any competitor $\uv_{1:2^d}^{(d)}\in V^{2^d}$, where we identify $2^d$ possible states by $1,\ldots,2^d$ and $T_s^{(d)}$ is the number of time steps with $s$ as side information.
\end{example}

While a larger $D$ can capture a longer dependence in the sequence, however, the performance of a per-state algorithm could significantly degrade due to the exponential number of states as pointed out in Remark~\ref{rem:effect_large_state}.

\subsubsection{Tree-Structured Side Information}\label{sec:tree_side_info}

The limitation of Markov side information motivates a general \emph{tree-structured side information} (or tree side information in short). 
Informally, we say that a sequence has a \emph{depth-$D$ tree structure} if the state at time $t$ depends on at most $D$ of the previous occurrences, corresponding to a full binary tree of depth $D$; see Figure~\ref{fig:suffix_tree}.
This degree of freedom allows to consider different lengths of history for each state, leading to the terminology \emph{variable-order Markov structure}, as opposed to the previous \emph{fixed-order Markov structure}.
If an underlying structure is approximately captured by a tree structure of depth $D$ with the number of leaves far fewer than $2^D$, the corresponding per-state algorithm can enjoy a much lower regret guarantee.



We now formally define a tree side information.
We say that a string $\omega_{1-l}\omega_{2-l}\ldots \omega_0$ is a \emph{suffix} of a string $\omega_{1-l'}'\omega_{2-l'}'\ldots \omega_0'$, if $l\le l'$ and $\omega_{-i}=\omega_{-i}'$ for all $i\in\{0,\ldots,l-1\}$.
Let $\lambda$ denote the empty string.
We define a \emph{(binary) suffix set} $\Tree$ as a set of binary strings that satisfies the following two properties~\citep{Willems--Shtarkov--Tjalkens1995}:
(1) Properness: no string in $\Tree$ is a suffix of any other string in $\Tree$; 
(2) Completeness: every semi-infinite binary string $\ldots h_{t-2}h_{t-1}h_t$ has a suffix from $\Tree$.
Since there exists an one-to-one correspondence between a binary suffix set and a full binary tree, we also call $\Tree$ a \emph{suffix tree}.
Given $D \ge 0$, let $\Tscr_{\le D}$ denote the set of all suffix trees of depth at most $D$. 

For a suffix tree $\Tree\in\Tscr_{\le D}$, we define a \emph{tree side information} $H_{\Tree;\Omega}$ with respect to $\Tree$ and $\Omega=(\omega_t)_{t\ge 1}$ as the matching suffix from the auxiliary sequence.
We can also identify $h_t$, the tree side information defined by $\Tree$ at time $t$, with a unique leaf node $s_t^{\Tree} \in \Tree$. 
For example, if a suffix set $\Tree$ consists of all possible $2^D$ binary strings of length $D\ge 1$, then it boils down to the fixed-order Markov case $h_t=\omega_{t-D}^{t-1}$.

For a single tree $\Tree$, the goal is to keep the regret
\[
\Reg(\uv[\Tree];\gv^T) \defeq \sum_{t=1}^T \langle\gv_t, \wv_t - \uv^{\Tree}_{s_t^{\Tree}} \rangle
\]
small for any competitor $\uv[\Tree] \defeq (\uv_s^{\Tree})_{s \in \Tree}$. 
In the next two subsections, we aim to follow the performance of the \emph{best suffix tree} of depth at most $D$, or equivalently, to keep the worst regret $\max_{\Tree \in \Tscr_{\le D}} \Reg_{\Ac}(\uv[\Tree];\gv^t)$ small for any collection of competitors $(\uv[\Tree])_{\Tree \in \Tscr_{\le D}}$. 

\begin{remark}[Matching Lower Bound]
When the auxiliary sequence $\Omega$ is constructed from a binary quantizer $Q$ with the history $\gv^{t-1}$ as mentioned earlier, we can show an optimality of the per-state KT algorithm in Section~\ref{sec:olo_side_information} for a single tree by establishing a matching regret lower bound extending the technique of \citet[Theorem~5.12]{Orabona2019}; see Appendix~\ref{supp:sec:lower_bound}.
\end{remark}

Below, we will use the \emph{tree potential} with respect to $\Tree$ and $\Omega$ defined as
\[
\Psikt(\gv^t;\Tree,\Omega)\defeq \prod_{s\in\Tree}\Psikt(\gv^t(s;\Omega)),
\]
where we write $s\in\Tree$ for any leaf node $s$ of the tree $\Tree$ with a slight abuse of notation and we define \[\gv^t(s;\Omega)\defeq (\gv_i\suchthat \textnormal{$s$ is a suffix of $\omega_{i-D}^{i-1}$, $1\le i\le t$}).\]
From now on, we will hide any dependence on $\Omega$ whenever the omission does not incur confusion.

\subsubsection{Context Tree Weighting for OLO with Tree Side Information}
\label{sec:ctw_olo}

To compete against the best competitor adaptive to \emph{any} tree side information of depth $\le D$, a natural solution 
is to consider a mixture of all tree potentials; note, however, that there are doubly-exponentially many $O(2^{2^{D}})$ possible suffix trees of depth $\le D$, and thus it is not computationally feasible to compute such a mixture naively.
Instead, inspired by the context tree weighting (CTW) probability assignment of \citet{Willems--Shtarkov--Tjalkens1995}, we analogously define the CTW potential as
$\Psictw(\gv^t)\defeq \Psictw_{\lambda}(\gv^t)$ with a recursive formula
\begin{align*}
\small
&\Psictw_s(\gv^t)
\numberthis
\label{eq:psictw_recursion}
\\&
\defeq\begin{cases}
\half\Psikt_s(\gv^t)
+\half\Psictw_{\bar{1} s}(\gv^t)\Psictw_{1 s}(\gv^t) & \textnormal{if }|s|<D\\
\Psikt_s(\gv^t) & \textnormal{if }|s|=D
\end{cases}
\end{align*}
\begin{wrapfigure}{r}{0.225\textwidth}
  \vspace{-1.5em}
  \begin{center}
    \includegraphics[width=.225\textwidth]{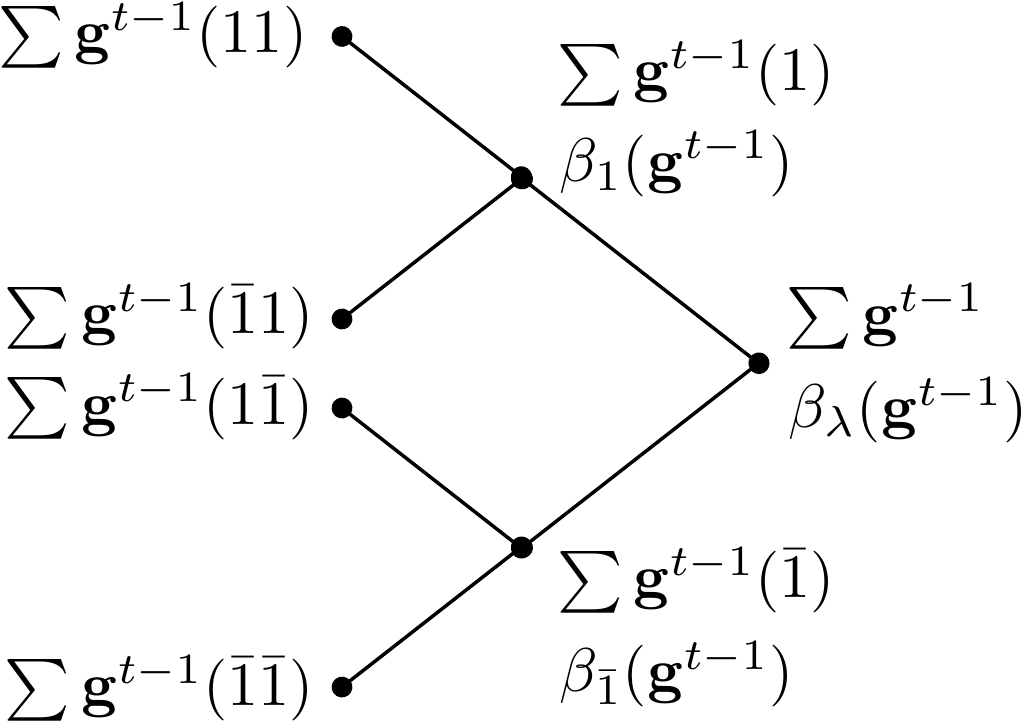}
  \end{center}
  \caption{A context tree of depth 2.}
  \label{fig:context_tree}
\end{wrapfigure}
for any binary string $s$ of length $\le D$ and $\Psikt_s(\gv^t)\defeq \Psikt(\gv^t(s))$.
Conceptually, this recursion can be performed over the perfect suffix tree of depth $D$, which we denote by $\Tc_D$ and call the context tree of depth $D$; see Figure~\ref{fig:context_tree} for the context tree of depth $D=2$.
Following the same logic of \citet{Willems--Shtarkov--Tjalkens1995}, one can easily show that
\[
\Psictw(\gv^t)=\sum_{\Tree\in\Tscr_{\le D}} w(\Tree)\Psikt(\gv^t;\Tree)
\]
for $w(\Tree)=2^{-\Gamma_D(\Tree)}$, where
$\Gamma_D(\Tree)\defeq 2|\Tree|-1-|\{s\in\Tree\suchthat |s|=D\}|$ is a complexity measure of a full binary tree $\Tree$ of depth $\le D$, 
$|\Tree|$ denotes the number of leaf nodes of a full binary tree $\Tree$, and $\Tscr_{\le D}$ denotes the set of all suffix trees of depth $\le D$.

For a path $\rho$ from the root to a leaf node of $\Tc_D$ and a full binary tree $\Tree$, we let $s_{\Tree}(\rho)$ denote the unique leaf node of $\Tree$ that intersects with the path $\rho$.
We also define $\vvkt(\gv^{t-1};\Tree)\defeq \vvkt(\gv^{t-1}(s_{\Tree}(\omega_{t-D}^{t-1})))$.
Then, based on the construction of the vectorial betting for a mixture potential in Section~\ref{sec:olo_multiple_side_information}, we define the vectorial CTW betting
\begin{align*}
\vvctw(\gv^{t-1})&\defeq \frac{\uvctw(\gv^{t-1})}{\Psictw(\gv^{t-1})},
\qquad\textnormal{where}
\numberthis
\label{eq:def_uvctw}
\\
\uvctw(\gv^{t-1})
&\defeq \sum_{\Tree\in\Tscr_{\le D}} w(\Tree)\Psikt(\gv^{t-1};\Tree) \vvkt(\gv^{t-1};\Tree),
\end{align*}
then we define the CTW OLO algorithm as the action
\[
\xvctw(\gv^{t-1})\defeq \vvctw(\gv^{t-1})\wealth_{t-1}(\gv^{t-1}).
\numberthis
\label{eq:ctw_olo}
\]
By Theorem~\ref{thm:olo_multiple_side_information}, we readily have the regret guarantee of the CTW OLO algorithm as follows:

\begin{corollary}
\label{cor:olo_ctw}
Let $D\ge 0$ be fixed.
For any $\gv_1,\ldots,\gv_T\in \unitball$, the CTW OLO algorithm~\eqref{eq:ctw_olo} satisfies
$\wealth_T\ge \wealth_0 \Psictw(\gv^T)$.
Moreover, we have
\begin{align*}
&\Reg(\uv[\Tree];\gv^T)\\
&= w(\Tree)\wealth_0 +  \tilde{O}\Biggl(\sqrt{\Bigl(\ln \frac{1}{w(\Tree)}\Bigr)\sum_{s\in\Tree} T_{s}^{\Tree}\|\uv_{s}^{\Tree}\|^2} \Biggr)
\end{align*}
for any tree $\Tree\in\Tscr_{\le D}$, where 
$T_{s}^{\Tree}$ denotes the number of occurrences of a side information symbol $s\in\Tree$ with respect to the tree side information $H_{\Tree;\Omega}$.
\end{corollary}

Hence, the CTW OLO algorithm~\eqref{eq:ctw_olo} can tailor to the best tree side information in hindsight.
Now, the remaining question is: can we \emph{efficiently} compute the vectorial CTW betting~\eqref{eq:def_uvctw}?
As a first attempt, the summation over the trees $\Tree\in\Tscr_{\le D}$ in \eqref{eq:def_uvctw} can be naively computed via a similar recursive formula as \eqref{eq:psictw_recursion}.
We define
\[\rho(\omega_{t-D}^{t-1})\defeq \{\lambda,\omega_{t-1},\ldots,\omega_{t-D}^{t-1}\}\] and call the \emph{active nodes} given the side information suffix $\omega_{t-D}^{t-1}$.

\begin{proposition}
\label{prop:ctw_update_validity}
For each node $s$ of $\Tc_D$, define
{\small
\begin{align*}
\uvctw_s(\gv^{t-1})
&\defeq
\begin{cases}
\half\Psikt_s(\gv^{t-1}) \vvkt_s(\gv^{t-1})\\ \quad+\half\uvctw_{\bar{1}s}(\gv^{t-1})\uvctw_{1s}(\gv^{t-1}) & \textnormal{if }|s|<D,\\
\Psikt_s(\gv^{t-1}) \vvkt_s(\gv^{t-1}) & \textnormal{if }|s|=D,
\end{cases}
\\
\vvkt_s(\gv^{t-1})
&\defeq 
\begin{cases}
\vvkt(\gv^{t-1}(s)) & \textnormal{if $s\in\rho(\omega_{t-D}^{t-1})$}\\
1 & \textnormal{otherwise.}
\end{cases}
\numberthis
\label{eq:uvctw_recursion}
\end{align*}}%
Then, the recursion is well-defined, and
$\uvctw_\lambda(\gv^{t-1})=\uvctw(\gv^{t-1})$.
\end{proposition}

While the recursions~\eqref{eq:psictw_recursion} and \eqref{eq:uvctw_recursion} take $O(2^D)$ steps for computing a mixture of $O(2^{2^D})$ many tree potentials, they are still not feasible as an online algorithm even for a moderate $D$.
In the next section, we show that the per-round time complexity $O(2^D)$ can be significantly improved to $O(D)$ by exploiting the tree structure further.

\subsubsection{The Efficient CTW OLO Algorithm with \texorpdfstring{$O(D)$}{O(D)} Steps Per Round}
\label{sec:ctw_olo_efficient_updates}

\paragraph{(1) Compute $\vvctw$ in $O(D)$ steps}
The key idea is that, given the suffix $\omega_{t-D}^{t-1}$, the vector betting $\vvctw=\uvctw/\Psictw$ can be computed efficiently via the recursive formulas~\eqref{eq:psictw_recursion} and \eqref{eq:uvctw_recursion}, by only traversing the active nodes $\rho(\omega_{t-D}^{t-1})= \{\lambda,\omega_{t-1},\ldots,\omega_{t-D}^{t-1}\}$ in the context tree $\Tc_D$.
In order to do so, we define
\[
\b_s(\gv^{t-1})\defeq \frac{\Psikt_s(\gv^{t-1})}{\Psictw_{\bar{1}s}(\gv^{t-1})\Psictw_{1s}(\gv^{t-1})}
\numberthis
\label{eq:def_beta}
\]
for every \emph{internal} node $s$ of $\Tc_D$.
\begin{restatable}{proposition}{PropVvctw}
\label{prop:vvctw}
Define
\begin{align*}
&\vvctw_{s_d}(\gv^{t-1})\\
&\defeq
\begin{cases}
\frac{\b_{s_d}(\gv^{t-1})}{\b_{s_d}(\gv^{t-1})+1} \vvkt_{s_d}(\gv^{t-1}) \\
\quad+ \frac{1}{\b_{s_d}(\gv^{t-1})+1} \vvctw_{s_{d+1}}(\gv^{t-1}) & \textnormal{if }d<D \\
\vvkt_{s_D}(\gv^{t-1}) & \textnormal{if }d=D
\end{cases}
\numberthis
\label{eq:vvctw_recursion}
\end{align*}
for $s_d=\omega_{t-d}^{t-1}\in\Tc_D$, $d=0,\ldots,D$. Then, $\vvctw(\gv^{t-1})=\vvctw_{\lambda}(\gv^{t-1})$.
\end{restatable}
Hence, if we can store $\sum\gv^{t-1}(s)$ and the value $\b_s(\gv^{t-1})$ as defined in \eqref{eq:def_beta} for every node $s$ of $\Tc_D$, we can compute $\vvctw$ in $O(D)$.

\paragraph{(2) Update \texorpdfstring{$\b_s$}{beta\_s} in $O(D)$ steps}
Upon receiving $\gv_t$, we need to update $\b_{s_d}(\gv^{t-1})$ as
\small
\[
\b_{s_d}(\gv^t)
=\b_{s_d}(\gv^{t-1}) \frac{\Psikt_{s_d}(\gv^t)}{\Psikt_{s_d}(\gv^{t-1})} \frac{\Psictw_{s_{d+1}}(\gv^{t-1})}{\Psictw_{s_{d+1}}(\gv^{t})}
\numberthis\label{eq:beta_update_recursion}
\]
\normalsize
for each $s_d=\omega_{t-d}^{t-1}\in\Tc_D$.
Here, the ratio $\Psictw_{s_{d}}(\gv^{t})/\Psictw_{s_{d}}(\gv^{t-1})$ can be also computed efficiently while traversing the path $\rho(\omega_{t-D}^{t-1})$ from the leaf node $s_D$ to the root $s_0=\lambda$, based on the following recursion:

\begin{proposition}
\label{prop:ctw_potential_ratio}
For each node $s_d=\omega_{t-d}^{t-1}\in\Tc_D$, $d=0,\ldots,D$,
\begin{align*}
&\frac{\Psictw_{s_d}(\gv^{t})}{\Psictw_{s_d}(\gv^{t-1})}\\
&=\begin{cases}
\frac{\b_{s_d}(\gv^{t-1})}{\b_{s_d}(\gv^{t-1})+1} \frac{\Psikt_{s_d}(\gv^t)}{\Psikt_{s_d}(\gv^{t-1})}\\
\quad+\frac{1}{\b_{s_d}(\gv^{t-1})+1}\frac{\Psictw_{s_{d+1}}(\gv^{t})}{\Psictw_{s_{d+1}}(\gv^{t-1})} & \textnormal{if }d<D\\
\frac{\Psikt_{s_D}(\gv^t)}{\Psikt_{s_D}(\gv^{t-1})} & \textnormal{if }d=D
\end{cases}.
\numberthis
\label{eq:ctw_potential_ratio_recursion}
\end{align*}
\end{proposition}

Hence, updating $\b_s$'s can be also performed efficiently in $O(D)$ time. 
The space complexity of this algorithm is $O(DT)$, since there can be at most $D$ nodes activated for the first time at each round.
The complete algorithm is summarized in Algorithm~\ref{alg:ctw_olo} in Appendix.

\section{EXPERIMENTS}
\label{sec:exps}

To validate the motivation of this work and demonstrate the power of the proposed algorithms in online convex optimization, we performed online linear regression with absolute loss following \citet{Orabona--Pal2016}.
We observed, however, that the datasets considered therein do not contain any temporal dependence and thus the proposed algorithms did not prove useful (data not shown).
Instead, we chose two real-world temporal datasets (Beijing PM2.5~\citep{Liang--Zou--Guo--Li--Zhang--Zhang--Huang--Chen2015} and Metro Interstate Traffic Volume~\citep{Hogue2019}) from the UCI machine learning repository~\citep{UCI2019}.
All details including data preprocessing can be found in Appendix~\ref{app:sec:exp} and the code that fully reproduce the results is available at \url{https://github.com/jongharyu/olo-with-side-information}.

To construct auxiliary sequences, we used the \emph{canonical binary quantizers $Q_{\ev_i}$}, where $\ev_i$ denotes the $i$-th standard vector.
We first ran the per-state versions of OGD, AdaNormal~\citep{McMahan--Orabona2014}, DFEG~\citep{Orabona2013}, and KT with Markov side information of different depths and ran the CTW algorithm for the maximum depth ranging $0,1,3\ldots,11$.
We optimally tuned the per-state OGD using only a single rate for all states due to the prohibitively large complexity of the optimal grid search; see Figures~\ref{fig:add_exps_metro}(a) and \ref{fig:add_exps_beijing}(a) in Appendix.   
While the per-state KT consistently showed the best performance, the performance degraded as we used too deep Markov side information beyond some threshold for all algorithms.
In Figures~\ref{fig:add_exps_metro}(b) and \ref{fig:add_exps_beijing}(b) in Appendix, CTW often achieved even better performance than the best performance achieved by KT across the different choices of quantizer, also being robust to the choice of the maximum depth. 

In practice, however, we do not know which dimension to quantize a priori. Hence, we showed the performance of the combined CTW algorithms over all $d$ quantizers aggregated by either the mixture or the addition---conceptually, the mixture of CTWs can be viewed as a \emph{context forest weighting}.
As a benchmark, we also ran the combined KT algorithms over all $d$ quantizers for each depth.
In Figure~\ref{fig:summary}, we summarized the per-coordinate results by taking the best performance over all quantizers; see the first five dashed lines in the legend. 
While these are only hypothetical which were not attained by an algorithm,
surprisingly, the combined CTW algorithms over different quantizers, either by the mixture or the addition of \citet{Cutkosky2019}, achieved the hypothetically best performance (plotted solid).

\begin{figure}[!htb]
    \centering
    \includegraphics[width=.425\textwidth]{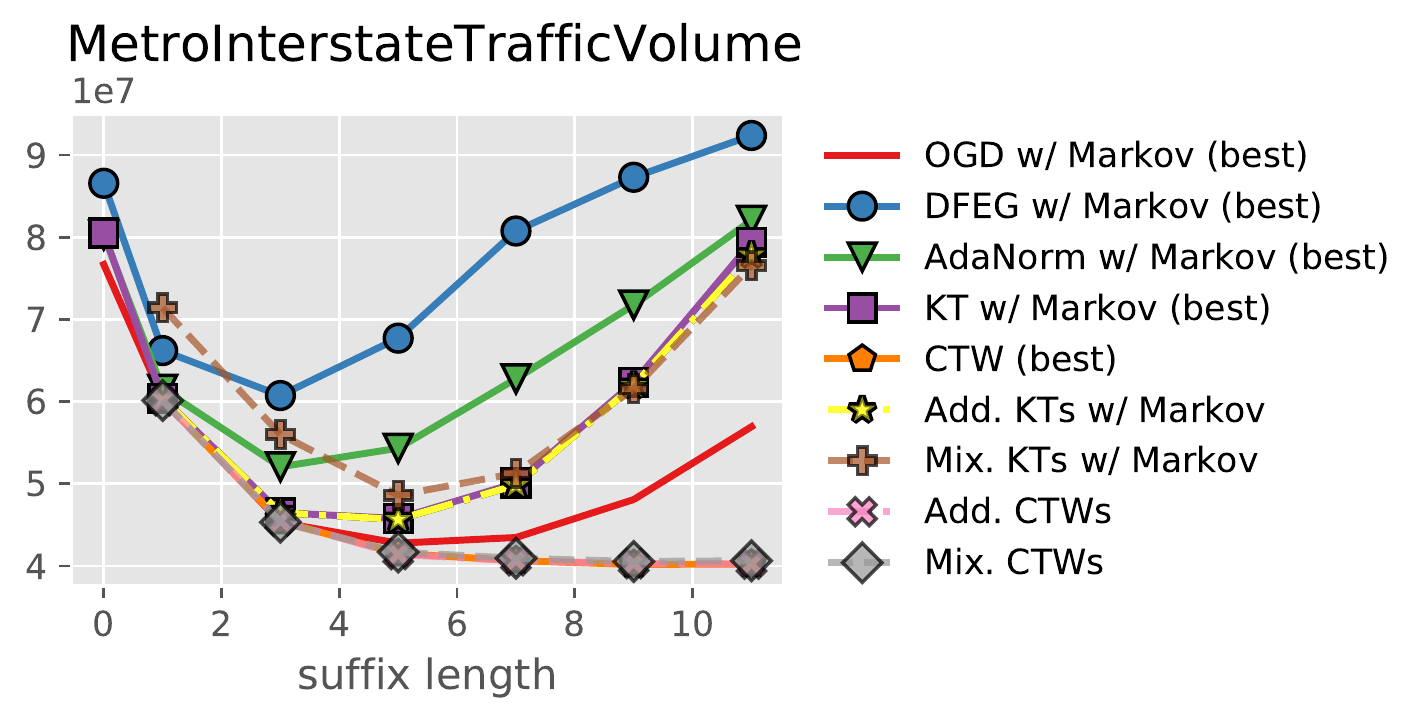}
    \vspace{-.5em}
    \includegraphics[width=.425\textwidth]{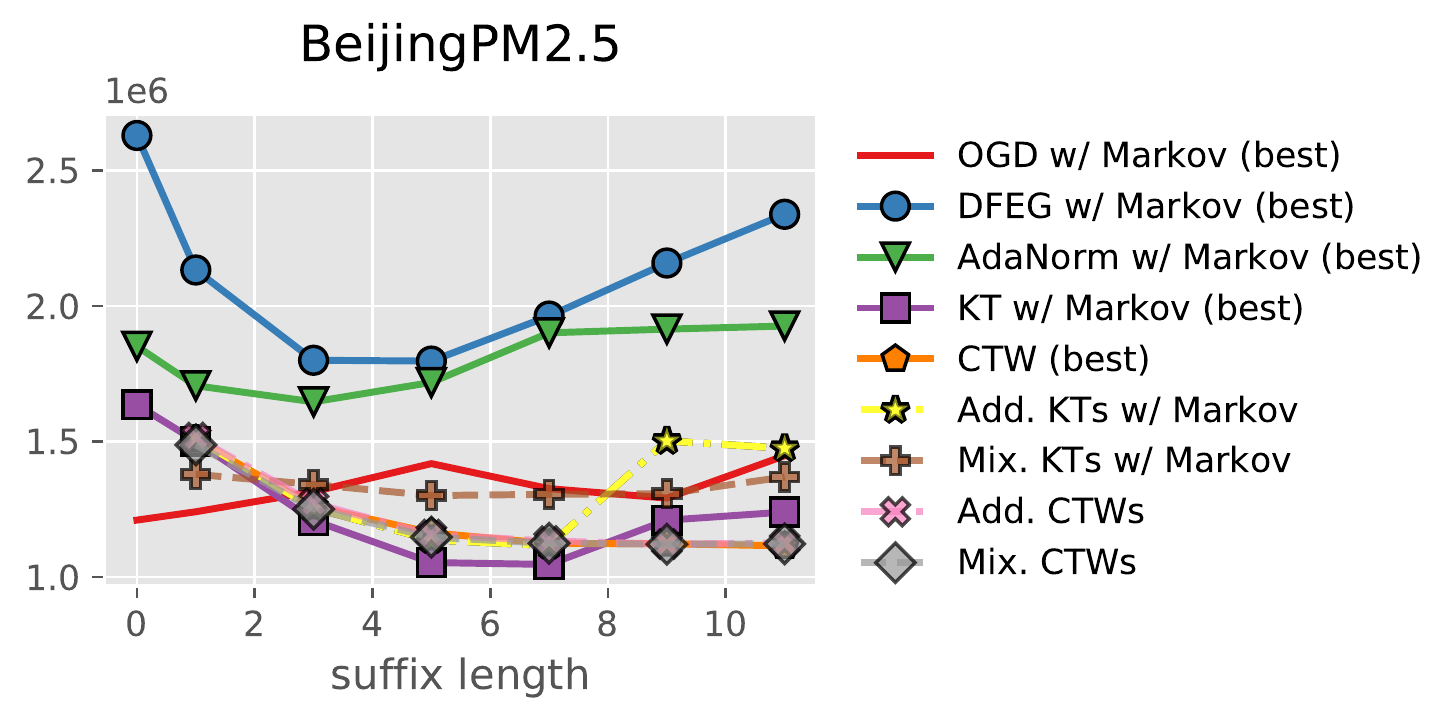}
    \caption{Summary of the experiments. 
    }
    \label{fig:summary}
\end{figure}

\newpage

\section{CONCLUDING REMARKS}
\label{sec:conclusion}
Aiming to leverage a temporal structure in the sequence $\gv^n$, we developed the CTW OLO algorithm that can efficiently adapt to the best tree side information in hindsight by combining a universal coin betting based OLO algorithm and universal compression (or prediction) techniques from information theory. Experimental results demonstrate that the proposed framework can be effective in solving real-life online convex optimization problems.

The key technical contribution of the paper is to consider the product and mixture potentials, motivated from information theory, and to adapt the CTW algorithm of \citet{Willems--Tjalkens--Ignatenko2006} to online linear optimization in Hilbert spaces. 
Main technical difficulties lie in analyzing the product potential (Proposition~\ref{prop:conjugate_base_function_product}) and properly invoking Rissanen's lower bound in Theorem~\ref{thm:olo_side_information_lower_bound} to establish the optimality.

We remark that an anonymous reader of an earlier version of this manuscript proposed a simpler alternative approach based on a meta algorithm that recasts any parameter-free OLO algorithm for tree-structured side information. 
The idea is to combine the specialist framework of \citet{Freund--Shapire--Singer--Warmuth1997} and apply the addition technique of \citet{Cutkosky2019}. 
Running a base OLO algorithm at each node of a context tree as a specialist, the meta algorithm adds up the outputs of the specialists on the active path at each round and updates them at the end of the round. 
This approach achieves a similar regret guarantee of the CTW OLO (Corllary~\ref{cor:olo_ctw}) with the same complexity. 
A detailed study is beyond the scope of this paper and thus left as future work. 


\newpage
\subsubsection*{Acknowledgements}
This work was supported in part by the National Science Foundation under Grant CCF-1911238.
The authors appreciate insightful feedback from anonymous reviewers to improve earlier versions of the manuscript.

\bibliographystyle{apalike}
\bibliography{ref}

\newcommand{\noopsort}[1]{}
\begin{thebibliography}{}

\bibitem[Bauschke and Combettes, 2011]{Bauschke--Combettes2011}
Bauschke, H.~H. and Combettes, P.~L. (2011).
\newblock {\em Convex analysis and monotone operator theory in Hilbert spaces},
  volume 408.
\newblock Springer.

\bibitem[Begleiter et~al., 2004]{Begleiter--El-Yaniv--Yona2004}
Begleiter, R., El-Yaniv, R., and Yona, G. (2004).
\newblock On prediction using variable order {M}arkov models.
\newblock {\em J. Artif. Intell. Res.}, 22:385--421.

\bibitem[Bhaskara et~al.,
  2020a]{Bhaskara--Cutkosky--Kumar--Purohit2020Imperfect}
Bhaskara, A., Cutkosky, A., Kumar, R., and Purohit, M. (2020a).
\newblock Online learning with imperfect hints.
\newblock In {\em Proc. Int. Conf. Mach. Learn.}, pages 822--831. PMLR.

\bibitem[Bhaskara et~al., 2020b]{Bhaskara--Cutkosky--Kumar--Purohit2020Many}
Bhaskara, A., Cutkosky, A., Kumar, R., and Purohit, M. (2020b).
\newblock Online linear optimization with many hints.
\newblock {\em arXiv preprint arXiv:2010.03082}.

\bibitem[Cesa-Bianchi and Lugosi, 2006]{Cesa-Bianchi--Lugosi2006}
Cesa-Bianchi, N. and Lugosi, G. (2006).
\newblock {\em Prediction, learning, and games}.
\newblock Cambridge University Press.

\bibitem[Chaudhuri et~al., 2009]{Chaudhuri--Freund--Hsu2009}
Chaudhuri, K., Freund, Y., and Hsu, D. (2009).
\newblock A parameter-free hedging algorithm.
\newblock In {\em Adv. Neural Inf. Proc. Syst.}, volume~22. Curran Associates,
  Inc.

\bibitem[Chen et~al., 2021]{Chen--Luo--Wei2021}
Chen, L., Luo, H., and Wei, C.-Y. (2021).
\newblock Impossible tuning made possible: A new expert algorithm and its
  applications.
\newblock {\em arXiv preprint arXiv:2102.01046}.

\bibitem[Chernov and Vovk, 2010]{Chernov--Vovk2010}
Chernov, A. and Vovk, V. (2010).
\newblock Prediction with advice of unknown number of experts.
\newblock In {\em Proc. Uncertain. Artif. Intell.}

\bibitem[Cover and Thomas, 2006]{Cover--Thomas2006}
Cover, T.~M. and Thomas, J.~A. (2006).
\newblock {\em Elements of information theory}.
\newblock John Wiley \& Sons.

\bibitem[Cutkosky, 2019]{Cutkosky2019}
Cutkosky, A. (2019).
\newblock Combining online learning guarantees.
\newblock In {\em Conf. Learn. Theory}, pages 895--913. PMLR.

\bibitem[Cutkosky and Boahen, 2017]{Cutkosky--Boahen2017}
Cutkosky, A. and Boahen, K. (2017).
\newblock Online learning without prior information.
\newblock In {\em Conf. Learn. Theory}, pages 643--677. PMLR.

\bibitem[Dekel et~al., 2017]{Dekel--Flajolet--Haghtalab--Jaillet2017}
Dekel, O., Flajolet, A., Haghtalab, N., and Jaillet, P. (2017).
\newblock Online learning with a hint.
\newblock In {\em Adv. Neural Inf. Proc. Syst.}, volume~30, pages 5299--5308.
  Curran Associates, Inc.

\bibitem[Dua and Graff, 2019]{UCI2019}
Dua, D. and Graff, C. (2019).
\newblock {UCI} {M}achine {L}earning {R}epository.

\bibitem[Duchi et~al., 2011]{Duchi--Hazan--Singer2011}
Duchi, J., Hazan, E., and Singer, Y. (2011).
\newblock Adaptive subgradient methods for online learning and stochastic
  optimization.
\newblock {\em J. Mach. Learn. Res.}, 12(7).

\bibitem[Foster et~al., 2015]{Foster--Rakhlin--Sridharan2015}
Foster, D.~J., Rakhlin, A., and Sridharan, K. (2015).
\newblock Adaptive online learning.
\newblock In {\em Adv. Neural Inf. Proc. Syst.}, volume~28, pages 3375--3383.
  Curran Associates, Inc.

\bibitem[Freund and Schapire, 1997]{Freund--Schapire1997}
Freund, Y. and Schapire, R.~E. (1997).
\newblock A decision-theoretic generalization of on-line learning and an
  application to boosting.
\newblock {\em Journal of Computer and System Sciences}, 55(1):119--139.

\bibitem[Freund et~al., 1997]{Freund--Shapire--Singer--Warmuth1997}
Freund, Y., Schapire, R.~E., Singer, Y., and Warmuth, M.~K. (1997).
\newblock Using and combining predictors that specialize.
\newblock In {\em Proc. Annu. ACM Symp. Theory Comput.}, pages 334--343.

\bibitem[Hogue, 2019]{Hogue2019}
Hogue, J. (2019).
\newblock Metro interstate traffic volume data set.

\bibitem[Jiao et~al., 2013]{Jiao--Permuter--Zhao--Kim--Weissman2013}
Jiao, J., Permuter, H.~H., Zhao, L., Kim, Y.-H., and Weissman, T. (2013).
\newblock Universal estimation of directed information.
\newblock {\em {IEEE} Trans. Inf. Theory}, 59(10):6220--6242.

\bibitem[Jun and Orabona, 2019]{Jun--Orabona2019}
Jun, K.-S. and Orabona, F. (2019).
\newblock Parameter-free online convex optimization with sub-exponential noise.
\newblock In {\em Conf. Learn. Theory}, pages 1802--1823. PMLR.

\bibitem[Jun et~al., 2017]{Jun--Orabona--Wright--Willett2017}
Jun, K.-S., Orabona, F., Wright, S., and Willett, R. (2017).
\newblock Online learning for changing environments using coin betting.
\newblock {\em Electron. J. Stat.}, 11(2):5282--5310.

\bibitem[Kelly~Jr., 1956]{Kelly1956}
Kelly~Jr., J.~L. (1956).
\newblock A new interpretation of information rate.
\newblock {\em {IRE} Trans. Inf. Theory}, 3(2):185--189.

\bibitem[Koolen and Van~Erven, 2015]{Koolen--VanErven2015}
Koolen, W.~M. and Van~Erven, T. (2015).
\newblock Second-order quantile methods for experts and combinatorial games.
\newblock In {\em Conf. Learn. Theory}, pages 1155--1175. PMLR.

\bibitem[Kozat et~al., 2008]{Kozat--Singer--Bean2008}
Kozat, S.~S., Singer, A.~C., and Bean, A.~J. (2008).
\newblock Universal portfolios via context trees.
\newblock In {\em Proc. {IEEE} Int. Conf. Acoust. Speech. Signal Process.},
  pages 2093--2096. IEEE.

\bibitem[Krichevsky and Trofimov, 1981]{Krichevsky--Trofimov1981}
Krichevsky, R. and Trofimov, V. (1981).
\newblock The performance of universal encoding.
\newblock {\em {IEEE} Trans. Inf. Theory}, 27(2):199--207.

\bibitem[Kuzborskij and Cesa-Bianchi, 2020]{Kuzborskij--Cesa-Bianchi2020}
Kuzborskij, I. and Cesa-Bianchi, N. (2020).
\newblock Locally-adaptive nonparametric online learning.
\newblock In {\em Adv. Neural Inf. Proc. Syst.}, volume~33.

\bibitem[Liang et~al.,
  2015]{Liang--Zou--Guo--Li--Zhang--Zhang--Huang--Chen2015}
Liang, X., Zou, T., Guo, B., Li, S., Zhang, H., Zhang, S., Huang, H., and Chen,
  S.~X. (2015).
\newblock Assessing {B}eijing's {PM2.5} pollution: Severity, weather impact,
  {APEC} and winter heating.
\newblock {\em Proc. R. Soc. A}, 471(2182):20150257.

\bibitem[Luo and Schapire, 2015]{Luo--Schapire2015}
Luo, H. and Schapire, R.~E. (2015).
\newblock Achieving all with no parameters: {A}da{N}ormal{H}edge.
\newblock In {\em Conf. Learn. Theory}, pages 1286--1304. PMLR.

\bibitem[McMahan and Abernethy, 2013]{McMahan--Abernethy2013}
McMahan, H.~B. and Abernethy, J. (2013).
\newblock Minimax optimal algorithms for unconstrained linear optimization.
\newblock In {\em Adv. Neural Inf. Proc. Syst.}, volume~26. Curran Associates,
  Inc.

\bibitem[McMahan and Orabona, 2014]{McMahan--Orabona2014}
McMahan, H.~B. and Orabona, F. (2014).
\newblock Unconstrained online linear learning in {H}ilbert spaces: Minimax
  algorithms and normal approximations.
\newblock In {\em Conf. Learn. Theory}, pages 1020--1039. PMLR.

\bibitem[Messias and Whiteson, 2018]{Messias--Whiteson2018}
Messias, J.~V. and Whiteson, S. (2018).
\newblock Dynamic-depth context tree weighting.
\newblock In {\em Adv. Neural Inf. Proc. Syst.}, volume~31. Curran Associates,
  Inc.

\bibitem[Orabona, 2013]{Orabona2013}
Orabona, F. (2013).
\newblock Dimension-free exponentiated gradient.
\newblock In {\em Adv. Neural Inf. Proc. Syst.}, volume~26, pages 1806--1814.
  Curran Associates, Inc.

\bibitem[Orabona, 2014]{Orabona2014}
Orabona, F. (2014).
\newblock Simultaneous model selection and optimization through parameter-free
  stochastic learning.
\newblock {\em arXiv preprint arXiv:1406.3816}.

\bibitem[Orabona, 2019]{Orabona2019}
Orabona, F. (2019).
\newblock A modern introduction to online learning.
\newblock {\em arXiv preprint arXiv:1912.13213}.

\bibitem[Orabona and Cutkosky, 2020]{Orabona--Cutkosky2020}
Orabona, F. and Cutkosky, A. (2020).
\newblock {ICML} 2020 tutorial on parameter-free online optimization.
\newblock Websites: \url{https://parameterfree.com/icml-tutorial/},
  \url{https://icml.cc/Conferences/2020/Schedule?showEvent=5753}.

\bibitem[Orabona and P{\'a}l, 2016]{Orabona--Pal2016}
Orabona, F. and P{\'a}l, D. (2016).
\newblock Coin betting and parameter-free online learning.
\newblock In {\em Adv. Neural Inf. Proc. Syst.}, volume~29. Curran Associates,
  Inc.

\bibitem[Orabona and Tommasi, 2017]{Orabona--Tommasi2017}
Orabona, F. and Tommasi, T. (2017).
\newblock Training deep networks without learning rates through coin betting.
\newblock In {\em Adv. Neural Inf. Proc. Syst.}, volume~30. Curran Associates,
  Inc.

\bibitem[Rakhlin and Sridharan, 2013]{Rakhlin--Sridharan13}
Rakhlin, A. and Sridharan, K. (2013).
\newblock Online learning with predictable sequences.
\newblock In {\em Conf. Learn. Theory}, pages 993--1019. PMLR.

\bibitem[Rissanen, 1984]{Rissanen1984}
Rissanen, J. (1984).
\newblock Universal coding, information, prediction, and estimation.
\newblock {\em {IEEE} Trans. Inf. Theory}, 30(4):629--636.

\bibitem[Rissanen, 1996]{Rissanen1996}
Rissanen, J.~J. (1996).
\newblock Fisher information and stochastic complexity.
\newblock {\em {IEEE} Trans. Inf. Theory}, 42(1):40--47.

\bibitem[Shalev-Shwartz, 2011]{Shalev-Shwartz2011}
Shalev-Shwartz, S. (2011).
\newblock Online learning and online convex optimization.
\newblock {\em Found. Trends Mach. Learn.}, 4(2):107--194.

\bibitem[Van~der Hoeven et~al., 2018]{VanDerHoeven--VanErven--Kotlowski}
Van~der Hoeven, D., van Erven, T., and Kot{\l}owski, W. (2018).
\newblock The many faces of exponential weights in online learning.
\newblock In {\em Conf. Learn. Theory}, pages 2067--2092. PMLR.

\bibitem[Willems et~al., 1995]{Willems--Shtarkov--Tjalkens1995}
Willems, F.~M., Shtarkov, Y.~M., and Tjalkens, T.~J. (1995).
\newblock The context-tree weighting method: Basic properties.
\newblock {\em {IEEE} Trans. Inf. Theory}, 41(3):653--664.

\bibitem[Willems et~al., 2006]{Willems--Tjalkens--Ignatenko2006}
Willems, F.~M., Tjalkens, T.~J., and Ignatenko, T. (2006).
\newblock Context-tree weighting and maximizing: {P}rocessing betas.
\newblock In {\em Proc. {UCSD} Inf. Theory Appl. Workshop}.

\bibitem[Xie and Barron, 1997]{Xie--Barron1997}
Xie, Q. and Barron, A.~R. (1997).
\newblock Minimax redundancy for the class of memoryless sources.
\newblock {\em {IEEE} Trans. Inf. Theory}, 43(2):646--657.

\bibitem[Zhang et~al., 2021]{Zhang--Wang--Yi--Yang2021}
Zhang, L., Wang, G., Yi, J., and Yang, T. (2021).
\newblock A simple yet universal strategy for online convex optimization.
\newblock {\em arXiv preprint arXiv:2105.03681}.

\bibitem[Ziv and Lempel, 1977]{Lempel--Ziv1977}
Ziv, J. and Lempel, A. (1977).
\newblock A universal algorithm for sequential data compression.
\newblock {\em {IEEE} Trans. Inf. Theory}, 23(3):337--343.

\end{thebibliography}

\newpage
\iflong
\onecolumn
\makesupplementtitle

\renewcommand{\thefigure}{\thesection.\arabic{figure}}    
\setcounter{table}{0}
\renewcommand{\thetable}{\thesection.\arabic{table}}
\setcounter{algorithm}{0}
\renewcommand{\thealgorithm}{\thesection.\arabic{algorithm}}

\appendix

\section{RELATED WORK}
\label{sec:related_work}
There have been several parameter-free methods proposed for OLO in Hilbert space~\citep{Orabona2013,Orabona2014,McMahan--Orabona2014,Orabona--Pal2016} as well as learning with expert advice (LEA)~\citep{Freund--Schapire1997,Chaudhuri--Freund--Hsu2009,Chernov--Vovk2010,Luo--Schapire2015,Foster--Rakhlin--Sridharan2015,Koolen--VanErven2015,Orabona--Pal2016}; see also \citep[Chapter 9]{Orabona2019} and the references therein. 
A parallel line of work on parameter-free methods considers the case when the maximum norm of $\gv_t$ (often referred to as the \emph{Lipschitz constant}), which is assumed to be 1 throughout in this paper, is unknown but the competitor norm $\|\uv\|$ is known~\citep{Duchi--Hazan--Singer2011, Cutkosky--Boahen2017}.
Recently, \citet{Zhang--Wang--Yi--Yang2021, Chen--Luo--Wei2021} studied a similar setting in this paper, albeit establishing guarantees only for bounded domains. 
We remark that AdaNormalHedge~\citep{Luo--Schapire2015} is a parameter-free LEA algorithm which can compete with mixtures of forcasters with side information, in particular tree experts via mixtures of sleeping experts; for example, \citet{Kuzborskij--Cesa-Bianchi2020} used AdaNormalHedge with tree experts for binary classification with absolute loss.
For a comprehensive overview of these parameter-free methods, see the tutorial~\citep{Orabona--Cutkosky2020}.

The connection between OLO and gambling was shown by \citet{Orabona--Pal2016}, where they also described a reduction for LEA.
This idea was also applied to training deep neural networks~\citep{Orabona--Tommasi2017}.
While the proposed algorithms in this paper are against \emph{stationary} competitors,
\citet{Jun--Orabona--Wright--Willett2017} proposed a coin betting based OLO algorithm against nonstationary competitors characterized by a sequence of vectors $\uv_1, \dotsc, \uv_T$ such that have at most $m$ change points. \citet[Section 5 and particularly Theorem 9]{VanDerHoeven--VanErven--Kotlowski} establishes a connection between the exponential weights (EW) algorithm and the coin-betting scheme. Earlier on in the paper, in Section 2 the interpretation of compression as a special case of EW with $\eta = 1$ is provided as well. Similarly, \citet{Jun--Orabona2019} utilize such a connection as well. To the best of our knowledge, however, we did not find a clear bridge constructed between compression and coin-betting methods in either, even though a careful examination of the mathematical details may hint toward this connection.

Universal compression, which is a classical topic in information theory, aims to compress sequences with no (or very little) statistical assumptions.
In the last century, there have been several techniques proposed that can compete against the best \iid compressor~\citep{Krichevsky--Trofimov1981, Rissanen1984, Xie--Barron1997}, finite state compressor~\citep{Lempel--Ziv1977} and tree compressor~\citep{Willems--Shtarkov--Tjalkens1995}. 
The CTW probability assignment invented by \citep{Willems--Shtarkov--Tjalkens1995} has been one of the most successful and widely used universal compression techniques. Beyond compression, this technique has been applied to estimation of directed information~\citep{Jiao--Permuter--Zhao--Kim--Weissman2013}, universal portfolios~\citep{Kozat--Singer--Bean2008}, and reinforcement learning~\citep{Messias--Whiteson2018}, to name a few. 
The efficient CTW OLO algorithm presented in Section~\ref{sec:ctw_olo_efficient_updates} is in the spirit of the processing betas algorithm proposed by \citet{Willems--Tjalkens--Ignatenko2006} for computing the predictive conditional probability induced by the CTW probability assignment~\citep{Willems--Shtarkov--Tjalkens1995}.
\citet[Section~5.3]{Cesa-Bianchi--Lugosi2006} also presented a CTW-based Hedge algorithm for LEA; see bibliographic remarks therein for other applications of CTW to learning problems. 

A related line of recent work on online learning with hints~\citep{Dekel--Flajolet--Haghtalab--Jaillet2017, Bhaskara--Cutkosky--Kumar--Purohit2020Imperfect, Bhaskara--Cutkosky--Kumar--Purohit2020Many} considers a scenario where the learner receives a vector $\hv_t$ with $\|\hv_t\| = 1$ such that $\langle \hv_t, \gv_t/\|\gv_t\| \rangle \ge \a > 0$ as a ``hint'' to the future. However, our setting is not directly comparable, since we only consider a finite side information and this line of work aims to establish small regret $o(\sqrt{T})$ measured with respect to static competitors.
We also remark that \citet{Rakhlin--Sridharan13} studied the problem of OLO when $\gv_t$ is modelled as a ``predictable'' sequence, in the sense that $\gv_t = M(\gv^{t-1}) + \nv_t$ with some adversarial noise $\nv_t$ with a (possibly randomized) function $M$; yet, they considered static competitors unlike this work. 


\newpage

\section{PER-STATE EXTENSIONS OF EXISTING ALGORITHMS}
\label{sec:other_per_state_algos}
Here we present per-state versions of OGD and two existing parameter-free OLO algorithms: the dimension-free exponentiated gradient algorithm (DFEG)~\citep{Orabona2013} and the adaptive normal algorithm (AdaNormal)~\citep{McMahan--Orabona2014}.

Following the original problem setting in \citep{Orabona2013}, we describe the per-state DFEG only for online linear regression. 
Consider a loss function $\ell(\yh,y)$, which is convex and $L$-Lipschitz in its first argument. At each round $t$, a learner picks $\wv_t\in V$. A nature then reveals $(\xv_t,y_t)\in V\times \Real$, and the learner suffers loss $\ell_t(\wv_t)\defeq \ell(\yh_t,y_t)$, where $\yh_t\defeq \langle \wv_t,\xv_t\rangle$.
Note that the DFEG algorithm requires a norm of the instance $\|\xv_t\|$ to form an action $\wv_t$. 

\newcommand{\thb}{\boldsymbol{\th}}
\begin{algorithm*}[htb]
\caption{Per-state Dimension-free Exponentiated Gradient~\citep{Orabona2013} for online regression}
\label{alg:DFEG}
\begin{algorithmic}[1]
\Procedure{PerStateDFEG}{$L, \delta, 0.882 \le a \le 1.109$}
   \State Initialize $\thb^{(s)} \gets 0 \in V, H^{(s)} \gets \delta$ for each $s\in[S]$
   \For{$1 \le t \le T$}
      \State Receive $h_t\in[S]$ and $\|\xv_t\|$
      \State Update $H^{(h_t)} \gets H^{(h_t)} + L^2 \max\{\|\xv_t\|,\|\xv_t\|^2\}$
      \State Set $\a_t \gets a (H^{(h_t)})^{1/2}, \b_t \gets (H^{(h_t)})^{3/2}$
      \If{$\|\thb^{(h_t)}\| =0$}
      \State Set $\wv_t \gets 0$
      \Else 
      \State Set $\wv_t \gets \frac{\thb^{(h_t)}}{\b_t \|\thb^{(h_t)}\|}\exp(\frac{\|\thb^{(h_t)}\|}{\a_t})$ 
      \EndIf
      \State Receive $(\xv_t,y_t)$ and incur loss $\ell_t(\wv_t)$ 
      \State Update $\thb^{(h_t)} \gets \thb^{(h_t)} - \partial \ell_t(\langle \wv_t, \xv_t \rangle )\xv_t$
   \EndFor
\EndProcedure
\end{algorithmic}
\end{algorithm*}

\begin{algorithm*}[htb]
\caption{Per-state AdaptiveNormal~\citep{McMahan--Orabona2014} for OLO with side information}\label{alg:AdaNorm}
\begin{algorithmic}[1]
\Procedure{PerStateAdaNormal}{$L, a \ge \frac{3L^2 \pi}{4}, \e$}
   \State Initialize $\thb^{(s)} \gets 0\in V$ for each $s\in[S]$
   \For{$1 \le t \le T$}
      \State Receive $h_t\in[S]$
      \If{$\|\thb^{(h_t)}\| = 0$}
      \State Set $\wv_t \gets 0$
      \Else 
      \State Set $\wv_t \gets \e \frac{\thb^{(h_t)}}{\|\thb^{(h_t)}\|} \frac{1}{2L \ln^2(t+1)} \{\exp(\frac{(\|\thb^{(h_t)}\|+L)^2}{2at}) - \exp(\frac{(\|\thb^{(h_t)}\|-L)^2}{2at})\}$ 
      \EndIf
      \State Receive $\gv_t$ and incur loss $\langle \gv_t, \wv_t \rangle $ 
      \State Update $\thb^{(h_t)} \gets \thb^{(h_t)} - \gv_t$
   \EndFor
\EndProcedure
\end{algorithmic}
\end{algorithm*}

We remark that these two algorithms are also guaranteed to incur essentially the same order of regret without tuning learning rate. 
Also, while the per-state KT OLO algorithm serves as a base algorithm in the CTW OLO algorithm, to be a fair comparison, the two algorithms can be also used as a base in the specialist framework to solve the tree side information problem, as noted in Section~\ref{sec:conclusion}.
There are, however, two minor disadvantages we can observe.
First of all, the DFEG algorithm is tailored to the online linear regression problem, while the per-state KT OLO and AdaptiveNormal algorithms can be applied to a general OLO problem. 
Second, while the KT OLO has only one hyperparameter, the initial wealth $\mathsf{W_0}$, the above two per-state algorithms have two hyperparameters (except the Lipschitz constant), which may need to be chosen or tuned in practice.

\newpage
\section{DEFERRED TECHNICAL MATERIALS}

\subsection{Proofs for Section~\ref{sec:prelim}}

\subsubsection{Proof of Theorem~\ref{thm:continuous_coin_betting_to_1D_OLO}}
We note that all statements in Section~\ref{sec:prelim} originally appeared in \citep{Orabona--Pal2016}. 
The proofs given here are rephrased and simplified from \citep{Orabona--Pal2016}.

Before we prove Theorem~\ref{thm:continuous_coin_betting_to_1D_OLO}, we state some key properties of the KT potential function $\psikt$.

\begin{proposition}
\label{prop:kt_potential}
For each $t\ge 1$ and any $g_1,\ldots,g_t\in[-1,1]$, the followings hold:
\begin{enumerate}[label=(\alph*)]
\item (Coordinatewise convexity) $g\mapsto \psikt(g^{t-1}g)$ is convex for $g\in[-1,1]$.
\item (Consistency) $\psikt(g^{t-1})=\half(\psikt(g^{t-1}1)+\psikt(g^{t-1}\bar{1}))$.
\item (The relation of signed betting and potential) \[\bkt(g^{t-1})=\frac{\psikt(g^{t-1}1)-\psikt(g^{t-1}\bar{1})}{\psikt(g^{t-1}1)+\psikt(g^{t-1}\bar{1})}=\frac{\psikt(g^{t-1}1)-\psikt(g^{t-1}\bar{1})}{\psikt(g^{t-1})}.\]
\item For any $x\in[0,t)$, $x(\psikt_t)''(x)\ge (\psikt_t)'(x)$.
\end{enumerate}
\end{proposition}
\begin{proof}
Recall $\qtkt_t(x)\defeq B(\frac{t+x+1}{2},\frac{t-x+1}{2})/B(\half,\half)$ and $\psikt(g^t)\defeq \psikt_t(\nsum g^t)\defeq 2^t\qtkt_t(\nsum g^t)$.
(a) and (d) follow from the properties of the Gamma function $\Gamma(\cdot)$; for details, see \cite[Lemma~12]{Orabona--Pal2016} and the proof therein.
(b) and (c) can be easily verified by the definition of the KT potential $\psikt$.
\end{proof}
We remark that the relation (b) can be understood as a continuous extension of the consistency of $\qtkt$ as a joint probability over a binary sequence $g^t\in\{-1,1\}^t$. Further, in view of the relation (c), the signed bet $\bkt$ is a continuous extension of the prequential probability $\qtkt(\cdot|g^{t-1})$ induced by the joint probability assignment $\qtkt(g^t)$.

We now show the following single round bound.
\begin{lemma}\label{lem:coin_betting_potential_implication}
For any $t\ge 1$ and $g_1,\ldots,g_t\in[-1,1]$, we have
\[
(1+g_t\bkt_t(g^{t-1}))\psikt(g^{t-1}) \ge \psikt(g^{t}).
\]
\end{lemma}
\begin{proof}
By the definition of coin betting potentials, we have
\begin{align*}
(1+g_t\bkt(g^{t-1}))\psikt(g^{t-1})
&\stackrel{(i)}{\ge} (1+g_t\bkt(g^{t-1}))\half(\psikt(g^{t-1}1) + \psikt(g^{t-1}\bar{1}))\\
&\stackrel{(ii)}{=}\Bigl(1+g_t\frac{\psikt(g^{t-1}1)-\psikt(g^{t-1}\bar{1})}{\psikt(g^{t-1}1) + \psikt(g^{t-1}\bar{1})}\Bigr)\half(\psikt(g^{t-1}1) + \psikt(g^{t-1}\bar{1}))\\
&= \frac{1+g_t}{2}\psikt(g^{t-1}1) + \frac{1-g_t}{2}\psikt(g^{t-1}\bar{1})\\
&\stackrel{(iii)}{\ge} \psikt(g^{t}).
\end{align*}
where $(i)$, $(ii)$, and $(iii)$ follow from (b), (c), and (a) in Proposition~\ref{prop:kt_potential}, respectively.
\end{proof}

While the above lemma establishes the lower bound on the cumulative wealth, 
we then need the following statement that connects regret and wealth via convex duality. 
We remark that this relation is the key statement that motivates all coin betting based algorithms.

\begin{proposition}[{\citealp{McMahan--Orabona2014}, \citep[Lemma~1]{Orabona--Pal2016}}]
\label{prop:regret_to_wealth}
Let $\Phi\colon V\to\Real$ be a convex function and let $\Phi^\star\colon V\to \Real\cup \{+\infty\}$ denote its Fenchel conjugate function.
For any $\gv_1,\ldots,\gv_T\in V^\star$ and 
any $\wv_t,\ldots,\wv_T\in V$, we have
\begin{align*}
\sup_{\uv\in V}
\{\Reg(\uv;\gv^T)-\Phi(\uv)\}
=-\sum_{t=1}^T \langle \gv_t,\wv_t\rangle
+\Phi^\star\Bigl(\sum_{t=1}^T \gv_t\Bigr),
\end{align*}
where $\Reg(\uv;\gv^T)\defeq \sum_{t=1}^T \langle \gv_t,\uv-\wv_t \rangle$.
\end{proposition}
\begin{proof}
By definition of Fenchel dual, we have
\begin{align*}
\sup_{\uv\in V}\{\Reg(\uv;\gv^T)-\Phi(\uv)\}
&= \sup_{\uv\in V}\Bigl\{\sum_{t=1}^T \langle \gv_t,\uv-\wv_t\rangle -\Phi(\uv)\Bigr\}\\
&= -\sum_{t=1}^T \langle \gv_t,\wv_t\rangle +  \sup_{\uv\in V}\Bigl\{\Bigl\langle \sum_{t=1}^T \gv_t ,\uv\Bigr\rangle -\Phi(\uv)\Bigr\}\\
&= -\sum_{t=1}^T \langle \gv_t,\wv_t\rangle
+\Phi^{\star}\Bigl(\sum_{t=1}^T \gv_t\Bigr).\qedhere
\end{align*}
\end{proof}

Now we are ready to prove Theorem~\ref{thm:continuous_coin_betting_to_1D_OLO}.

\begin{proof}[Proof of Theorem~\ref{thm:continuous_coin_betting_to_1D_OLO}]
We first show the wealth lower bound $\wealth_t\ge \wealth_0\psikt(g^t)$ stated in \eqref{eq:1d_coin_betting_kt_lower_bound} by induction on $t$. 
Suppose that $\wealth_{t-1}\ge \wealth_0\psikt(g^{t-1})$.
Then,
\begin{align*}
\wealth_t &= \wealth_{t-1} + g_t w_t\\
&= (1+\bkt(g^{t-1}) g_t)\wealth_{t-1}\\
&\stackrel{(a)}{\ge} (1+\bkt(g^{t-1}) g_t) \wealth_0\psikt(g^{t-1})\\
&\stackrel{(b)}{\ge} \wealth_0\psikt(g^{t}),
\end{align*}
where $(a)$ follows from the induction hypothesis and $(b)$ follows from Lemma~\ref{lem:coin_betting_potential_implication}.

The wealth lower bound can be converted into the desired regret bound by Proposition~\ref{prop:regret_to_wealth}.
That is, we have
\[
\sup_{u\in\Real}\{\Reg(u;g^T)-\phi(u)\}
= - \sum_{t=1}^T g_tw_t + \wealth_0\psikt(g^T) \le \wealth_0,
\]
where $\phi\suchthat\Real\to\Real$ is a convex function such that its conjugate function $\phi^\star\suchthat\Real\to\Real\cup\{+\infty\}$ is equal to $\wealth_0\psikt_T(\nsum g^t)$.
Since $x\mapsto \psikt_T(x)$ is a convex, proper, closed function, one can check that $\phi(u)=\wealth_0(\psikt_T)^\star(\frac{u}{\wealth_0})$ using Lemma~\ref{lem:dual_scale}.
\end{proof}

\subsubsection{Proof of Theorem~\ref{thm:continuous_coin_betting_to_Hilbert_OLO}}

As in 1D OLO case, we first show the following single round bound.
\begin{lemma}
For any $\gv_1,\ldots,\gv_t\in \unitball$, we have
\[
(1+\langle \gv_t,\vvkt(\gv^{t-1})\rangle ) \Psikt(\gv^{t-1}) \ge \Psikt(\gv^t).
\]
\end{lemma}

\begin{proof}
Let $\fv_{t-1}\defeq \nsum\gv^{t-1}$.
Consider
\begin{align*}
(1+&\langle \gv_t,\vvkt(\gv^{t-1})\rangle ) \Psikt(\gv^{t-1}) - \Psikt(\gv^t)\\
&= \Psikt(\gv^{t-1}) + 
\langle \gv_t,\vvkt(\gv^{t-1})\rangle \Psikt(\gv^{t-1}) - \Psikt(\gv^t)\\
&= \psikt_{t-1}(\|\fv_{t-1}\|) + \Bigl\langle \gv_t,\bkt_t(\|\fv_{t-1}\|)\frac{\fv_{t-1}}{\|\fv_{t-1}\|}\Bigr\rangle \psikt_{t-1}(\|\fv_{t-1}\|) - \psikt_t(\|\fv_{t-1}+\gv_t\|)\\
&\stackrel{(a)}{\ge} 
\psikt_{t-1}(\|\fv_{t-1}\|) + 
\min_{r\in\{\pm 1\}} \{
r\|\gv_t\| \bkt_t(\|\fv_{t-1}\|) \psikt_{t-1}(\|\fv_{t-1}\|) - \psikt_t(\|\fv_{t-1}\|+r\|\gv_t\|)
\}\\
&= \min_{r\in\{\pm 1\}} \{
(1+r\|\gv_t\| \bkt_t(\|\fv_{t-1}\|)) \psikt_{t-1}(\|\fv_{t-1}\|) - \psikt_t(\|\fv_{t-1}\|+r\|\gv_t\|)
\}\\
&\ge \min_{g\in[-1,1]} \{
(1+g\bkt_t(\|\fv_{t-1}\|)) \psikt_{t-1}(\|\fv_{t-1}\|) - \psikt_t(\|\fv_{t-1}\|+g)
\}\\
&\stackrel{(b)}{\ge} 0.
\end{align*}
Here, we apply Lemma~\ref{lem:extremes} since $\psikt_t$ satisfies $x(\psikt_t)''(x)\ge (\psikt_t)'(x)$ for all $x\in[0,t)$, to have $(a)$ by plugging in $\uv\gets \gv_t$, $\vv\gets \fv_{t-1}$, $c(\|\uv\|,\|\vv\|)\gets \frac{\bkt_t(\|\fv_{t-1}\|)}{\|\fv_{t-1}\|} \psikt_{t-1}(\|\fv_{t-1}\|)$, and $h(\cdot)\gets \psikt_t(\cdot)$. 
$(b)$ follows from the single round bound for 1D case established in Lemma~\ref{lem:coin_betting_potential_implication}.
\end{proof}

The proof of Theorem~\ref{thm:continuous_coin_betting_to_Hilbert_OLO} now follows similarly to that of Theorem~\ref{thm:continuous_coin_betting_to_1D_OLO}.
\begin{proof}[Proof of Theorem~\ref{thm:continuous_coin_betting_to_Hilbert_OLO}]
We show $\wealth_t\ge \wealth_0\Psikt(\gv^t)$
by induction on $t$.
For $t=0$, it trivially holds. 
For $t\ge 1$, assume that $\wealth_{t-1}\ge \wealth_0\Psikt(\gv^{t-1})$ holds.
Then, we have
\begin{align*}
\wealth_t &= \langle \gv_t, \xvkt_t\rangle + \wealth_{t-1}\\
&= (1+\langle \gv_t,\vvkt(\gv^{t-1})\rangle) \wealth_{t-1}\\
&\stackrel{(a)}{\ge} (1+\langle \gv_t,\vvkt(\gv^{t-1})\rangle) \wealth_0\Psikt(\gv^{t-1})\\
&\stackrel{(b)}{\ge} 
\wealth_0\Psikt(\gv^{t}).
\end{align*}
Here, $(a)$ follows from the induction hypothesis and $(b)$ follows from the above lemma.
The regret bound follows by the same logic of the 1D case using Proposition~\ref{prop:regret_to_wealth} with the additional application of Lemma~\ref{lem:dual_composition_dual}, which implies that
$(\psikt_t)^\star(\uv)=(\psikt_t)^\star(\|\uv\|)$.
\end{proof}

\subsection{Proofs for Section~\ref{sec:olo_side_information}}
\subsubsection{Proof of Theorem~\ref{thm:olo_side_information}}

The following statement generalizes Proposition~\ref{prop:regret_to_wealth} for static competitors to adaptive competitors.
\begin{proposition}\label{prop:regret_to_wealth_side_information}
Let $\Phi\colon V\times\cdots\times V\to\Real$ be a convex function and let $\Phi^\star\colon V\times\cdots\times V\to \Real\cup \{+\infty\}$.
For any side information sequence $H=(h_t)_{t\ge 1}$, any $\gv_1,\ldots,\gv_T\in V^\star$, and 
any $\wv_t,\ldots,\wv_T\in V$, we have
\begin{align*}
\sup_{\uv_{1:S}\in V\times\cdots V}
\{\Reg(\uv_{1:S}[H];\gv^T)-\Phi(\uv_{1:S})\}
=-\sum_{t=1}^T \langle \gv_t,\wv_t\rangle
+\Phi^\star\Bigl(\sum_{t\in[T]\suchthat h_t=1} \gv_t,\ldots,\sum_{t\in[T]\suchthat h_t=S} \gv_t\Bigr),
\end{align*}
where $\Reg(\uv_{1:S}[H];\gv^T)\defeq \sum_{s=1}^S \sum_{t\in[T]\suchthat h_t=S} \langle \gv_t,\uv_s-\wv_t \rangle$.
\end{proposition}
\begin{proof}
By definition of Fenchel dual, we have
\begin{align*}
\sup_{\uv_{1:S}\in V\times\cdots V}\{\Reg(\uv_{1:S}[H];\gv^T)-\Phi(\uv_{1:S})\}
&= \sup_{\uv_{1:S}\in V\times\cdots V}
\Bigl\{\sum_{s=1}^S\sum_{t\in[T]\suchthat h_t=s} \langle \gv_t,\uv_s-\wv_t\rangle -\Phi(\uv_{1:S})\Bigr\}\\
&= -\sum_{t=1}^T \langle \gv_t,\wv_t\rangle +  \sup_{\uv_{1:S}\in V\times\cdots V}
\Bigl\{\sum_{s=1}^S\Bigl\langle \sum_{t\in [T]\suchthat h_t=s} \gv_t ,\uv_s\Bigr\rangle -\Phi(\uv_{1:S})\Bigr\}\\
&= -\sum_{t=1}^T \langle \gv_t,\wv_t\rangle
+\Phi^\star\Bigl(\sum_{t\in[T]\suchthat h_t=1} \gv_t,\ldots,\sum_{t\in[T]\suchthat h_t=S} \gv_t\Bigr).
\qedhere
\end{align*}
\end{proof}

We are now ready to prove Theorem~\ref{thm:olo_side_information}.
\begin{proof}[Proof of Theorem~\ref{thm:olo_side_information}]
Since the vectorial betting $\vvkt(\gv^{t-1};h^t)$ only affects the component potential $\Psikt(\gv^t(h_t;h^{t-1}))$ by construction, the wealth lower bound readily follows from the same argument in the proof of Theorem~\ref{thm:continuous_coin_betting_to_Hilbert_OLO}.
Now, we observe that
\[
\Psikt(\gv^T;h^T) = 2^T \prod_{s\in[S]}\qtkt_{T_s}(\|\nsum \gv^T(s;h^T)\|),
\]
where $T_s\defeq |\{t\in[T]\suchthat h_t=s\}|$.
Since $\qtkt_T(x)\ge \frac{1}{2^T e\sqrt{\pi}}\frac{1}{\sqrt{T}}e^{\frac{2x^2}{T}}$ for $T\ge 1$ by {\citep[Lemma~14]{Orabona--Pal2016}}, we have
\[
\Psikt(\gv^T;h^T) \ge \Bigl(\frac{1}{e\sqrt{\pi}}\Bigr)^{S'} \frac{1}{\sqrt{T_1'\cdots T_S'}} \exp\Bigl(\sum_{s=1}^S \frac{2\|\nsum \gv^T(s;h^T)\|^2}{T_s'}\Bigr),
\]
where $S'\defeq \sum_{s=1}^S 1\{T_s\ge 1\}$ and $T_s'\defeq T_s\vee 1$.
Applying Propositions~\ref{prop:regret_to_wealth_side_information} and \ref{prop:conjugate_base_function_product} then establishes the regret upper bound.
\end{proof}

\subsubsection{Proof of Theorem~\ref{thm:olo_multiple_side_information}}
We show $\wealth_t\ge \wealth_0\psimix(\gv^t;\hv^t)$
by induction on $t$.
For $t=0$, it trivially holds. 
For $t\ge 1$, assume that $\wealth_{t-1}\ge \wealth_0\psimix(\gv^{t-1};\hv^{t-1})$ holds.
Then, we have
\begin{align*}
\wealth_t &= \langle \gv_t, \xvmix_t(\gv^{t-1};\hv^t)\rangle + \wealth_{t-1}\\
&= (1+\langle \gv_t,\vvmix(\gv^{t-1};\hv^t)\rangle) \wealth_{t-1}\\
&\stackrel{(a)}{\ge} (1+\langle \gv_t,\vvmix(\gv^{t-1};\hv^t)\rangle) \wealth_0\psimix(\gv^{t-1};\hv^{t-1})\\
&\stackrel{(b)}{\ge} 
\wealth_0\psimix(\gv^{t};\hv^t).
\end{align*}
Here, $(a)$ follows from the induction hypothesis, and $(b)$ follows from the construction of $\vvmix(\gv^{t-1};\hv^t)$.
The regret guarantee for $m\in[M]$ readily follows from the construction of the mixture potential, which guarantees $\wealth_T\ge w_m\wealth_0\Psikt(g^T;(h^{(m)})^T)$.
\qed

\subsubsection{Matching lower bounds for tree side information}
\label{supp:sec:lower_bound}
We first require the following theorem from \citep{Orabona2019}.
\begin{restatable}[{\citealp[Theorem~5.11]{Orabona2019}}]{theorem}{ThmOLOSuffCondCoinBetting}
\label{thm:olo_suff_cond_coin_betting}
Suppose that an OLO algorithm satisfies that for each $t\ge 0$
\[
\sup_{\gv^t\in\unitball^t}\Reg(\mathbf{0};\gv^t)
=-\inf_{\gv^t\in\unitball^t}\sum_{i=1}^t \langle \gv_i,\wv_i\rangle
\le \wealth_0^{(t)}
\numberthis\label{eq:mild_olo_algo_condition}
\]
with some nondecreasing sequence $(\wealth_0^{(t)})_{t\ge 0}$.
Then, for each $T\ge 1$, there exists $\vv_1,\ldots,\vv_T\in \unitball$ such that
\[
\wv_t=\vv_t\Bigl(\wealth_0^{(T)}+\sum_{i=1}^{t-1}\langle \gv_i,\wv_i\rangle\Bigr) \quad\text{for all $t\in [T]$}.
\]
\end{restatable}

For a binary quantizer $Q\suchthat \unitball\to \{\pm 1\}$, let $H_{\Tree,Q}$ denote the tree side information with respect to a tree $\Tree$ and an auxiliary sequence $\Omega=(\omega_t)_{t\ge 1}$ with $\omega_t=Q(\gv_t)$.
\begin{restatable}{theorem}{ThmOLOSideInformationLowerBound}
\label{thm:olo_side_information_lower_bound}
Let $V=\Real^d$ be the $d$-dimensional Euclidean space.
Suppose that a binary quantizer $Q\suchthat\unitball\to\{\pm 1\}$ satisfies $Q(\ev_j)=1$ and $Q(-\ev_j)=-1$ for some $j\in[d]$.
For $T$ sufficiently large, for any \emph{causal} OLO algorithm that satisfies the condition~\eqref{eq:mild_olo_algo_condition} in Theorem~\ref{thm:olo_suff_cond_coin_betting}, for any binary suffix tree $\Tree$, there exist a sequence $\gv_1,\ldots,\gv_T\in\unitball$ and a competitor $(\uv_s^*)_{s\in\Tree}[H_{\Tree,Q}]\in\Mc(H_{\Tree,Q})$ such that 
\begin{align*}
\Reg((\uv_s^*)_{s\in\Tree}[H_{\Tree,Q}]);\gv^T) 
&\ge 
\sqrt{
\sum_{s\in\Tree} T_s\|\uv_{s}^*\|_2^2 \ln\Bigl(
\frac{(T/|\Tree|)^{|\Tree|}}{(\wealth_0^{(T)})^2} \sum_{s\in\Tree}  T_s\|\uv_{s}^*\|_2^2 + 1\Bigr)}
+ \wealth_0^{(T)}.
\end{align*}
\end{restatable}

\begin{proof}
Without loss of generality, assume that the binary quantizer $Q\suchthat\unitball\to\{\pm 1\}$ satisfies $Q(\ev_1)=1$ and $Q(-\ev_1)=-1$.
For a binary sequence $c^T\in\{\pm 1\}^T$,
we set $\gv_t=(c_t,0,\ldots,0)$ for $c_t\in\{\pm 1\}$, so that $\langle \gv_t,\wv_t\rangle =c_t x_{t1}$.
Then, by Theorem~\ref{thm:olo_suff_cond_coin_betting}, we can write
\[
x_{t1}
=v_{t1}\Bigl(\wealth_0^{(T)} + \sum_{i=1}^{t-1} \langle \gv_i,\wv_i \rangle\Bigr)
=v_{t1}\Bigl(\wealth_0^{(T)} + \sum_{i=1}^{t-1} c_i x_{i1}\Bigr)
\]
for some $v_{t1}$ such that $|v_{t1}|\le 1$. 
Hence, the OLO problem with any causal algorithms satisfying~\eqref{eq:mild_olo_algo_condition} with respect to the 1D sequences $\gv^T$ can be equivalently viewed as the 1D coin betting with initial wealth $\wealth_0=\wealth_0^{(T)}$.

Now, we state the celebrated Rissanen's lower bound for universal compression in the form of the wealth upper bound for the coin betting. 
\citet{Rissanen1996} showed that for any probability assignment $q(x^T)$ on a binary sequence $x^T\in\{0,1\}^T$, there exists a sequence $\xt^T\in\{0,1\}$ such that
\[
q(\xt^T)\le e^{-\frac{|\Tree|}{2}\ln\frac{T}{|\Tree|}} \max_{p_{\Tree}} p_{\Tree}(\xt^T),
\]
where the maximum is over all possible tree sources $p_{\Tree}$ with the underlying tree $\Tree$.
This can be translated into the wealth upper bound for the standard coin betting with binary outcomes $c_t\in\{\pm 1\}$ thanks to the equivalence between the coin betting and universal compression: 
for any continuous coin betting algorithm which plays a relative bet $b_t\in[-1,1]$ at time $t$, there exists a binary sequence $\ct^T\in\{\pm 1\}^T$ such that
\begin{align*}
\frac{\wealth_T}{\wealth_0}=\prod_{t=1}^T (1+b_t \ct_t) 
&\le \Bigl(\frac{|\Tree|}{T}\Bigr)^{\frac{|\Tree|}{2}} \prod_{s\in\Tree}\max_{b_s\in[-1,1]} \prod_{t\in[T]\suchthat h_t=s} (1+b_s \ct_t)
\\
&\stackrel{(a)}{\le} \Bigl(\frac{|\Tree|}{T}\Bigr)^{\frac{|\Tree|}{2}} \prod_{s\in\Tree} \exp\Bigl(\frac{\ln 2}{T_s'} \Bigl(\sum_{t\in[T]\suchthat h_t=s} \ct_t\Bigr)^2\Bigr),\\
&= f\bigl((\nsum \ct^T(s;H_{\Tree,Q}))_{s\in\Tree}\bigr),
\numberthis\label{eq:tree_wealth_upper_bound_explicit}
\end{align*}
where $h_t$ denotes the suffix of the sequence $c^{t-1}$ with respect to $\Tree$ at time $t$, $T_s'\defeq T_s\vee 1$, $T_s\defeq |\{t\in[T]\suchthat h_t=s\}|$, $f((x_s)_{s\in\Tree})\defeq\prod_{s\in\Tree} h_s(x_s)$, and $h_s(x_s)=\b_s\exp(\frac{x_s^2}{2\a_s})$ with $\a_s=\frac{2T_{s'}}{\ln 2}$, and $\b_s=\sqrt{|\Tree|/T}$.
Here, $(a)$ follows by Lemma~\ref{lem:best_wealth_upper_bound}.

For the adversarial coin sequence $(\ct_t)_{t\ge 1}$ satisfying \eqref{eq:tree_wealth_upper_bound_explicit}, define $\gv_t\defeq (\ct_t,0,\ldots,0)$. Then, we have
\begin{align*}
\wealth_0^{(T)} + \sum_{t=1}^T \langle \tilde{\gv}_t,\wv_t\rangle
&= \wealth_0^{(T)} + \sum_{t=1}^T \ct_t x_{t1}\\
&\le \wealth_0^{(T)} f\bigl((\nsum \ct^T(s;H_{\Tree,Q}))_{s\in\Tree}\bigr)\\
&= \sum_{s\in\Tree} \bigl(\nsum \ct^T(s;H_{\Tree,Q})\bigr) u_s^* - \wealth_0^{(T)} f^\star\Bigl(\Bigl(\frac{|u_s^*|}{\wealth_0^{(T)}}\Bigr)_{s\in\Tree}\Bigr)\\
&= \sum_{t=1}^T \langle \gv_t,\uv_{h_t}^*\rangle - \wealth_0^{(T)} f^\star\Bigl(\Bigl(\frac{\|\uv_s^*\|_2}{\wealth_0^{(T)}}\Bigr)_{s\in\Tree}\Bigr),
\end{align*}
where $(u_s^*)_{s\in\Tree}=\wealth_0^{(T)} \nabla f((\nsum \ct^T(s;H_{\Tree,Q}))_{s\in\Tree})$ and $\uv_s^*\defeq (u_s^*,0,\ldots,0)$ for each $s\in\Tree$.
Rearranging the terms, we have
\begin{align*}
\Reg((\uv_s^*)_{s\in\Tree}[H_{\Tree,Q}]);\gv^T) 
&=\sum_{t=1}^T \langle \gv_t,\uv_{h_t}^*\rangle-\sum_{t=1}^T \langle \gv_t,\wv_t\rangle\\
&\ge \wealth_0^{(T)} + \wealth_0^{(T)} f^\star\Bigl(\Bigl(\frac{\|\uv_s^*\|_2}{\wealth_0^{(T)}}\Bigr)_{s\in\Tree}\Bigr).\qedhere
\end{align*}
\end{proof}

\subsubsection{Proof of Proposition~\ref{prop:ctw_update_validity}}

We use a backward induction over the depth $|s|$ to show that the recursion is well-defined.
First, if $|s|=D$, $\uvctw_s(\gv^{t-1})=\Psikt_s(\gv^{t-1}) \vvkt_s(\gv^{t-1})$. 
By definition of $\vvkt_s(\gv^{t-1})$, $\uvctw_s(\gv^{t-1})$ is a vector if $s$ is the active node at depth $D$, and a scalar otherwise.
Now, for $d\le D-1$, assume that $\uvctw_{s'}(\gv^{t-1})$ is a scalar if $s'$ is an active node and a vector otherwise for any $|s'|=d+1$ (induction hypothesis). Consider any node $s$ of $\Tc_D$ with $|s|=d$. 
If $s$ is an active node, then $\uvctw_{\bar{1}s}(\gv^{t-1})\uvctw_{1s}(\gv^{t-1})$ is a vector by the induction hypothesis, since exactly one of $\bar{1} s$ and $1s$ is active. Hence, $\uvctw_{s}(\gv^{t-1})$ is a vector.
If $s$ is not an active node, then, $\uvctw_{\bar{1}s}(\gv^{t-1})\uvctw_{1s}(\gv^{t-1})$ is a scalar by the induction hypothesis, since neither of $\bar{1} s$ and $1s$ is active. Hence, $\uvctw_{s}(\gv^{t-1})$ is a scalar. This completes the induction and thus the recursion is well-defined for all nodes $s$.

The claim $\uvctw_\lambda(\gv^{t-1})=\uvctw(\gv^{t-1})$ can be checked by a similar induction argument.
\qed

\subsubsection{Proof of Proposition~\ref{prop:vvctw}}
We claim that $\vvctw_s(\gv^{t-1})=\frac{\uvctw_s(\gv^{t-1})}{\Psictw_s(\gv^{t-1})}$ for any $s=s_d=\omega_{t-d}^{t-1}\in\Tc_D$, $d=0,\ldots,D$.
This trivially holds for the leaf node $s_D=\omega_{t-D}^{t-1}$.
For the internal nodes $s_d$ with $d<D$,
by plugging in the recursive formulas of $\uvctw(\gv^{t-1})$ and $\Psictw(\gv^{t-1})$, we can write
\[
\frac{\uvctw(\gv^{t-1})}{\Psictw(\gv^{t-1})}
=\frac{\b_s(\gv^{t-1})}{\b_s(\gv^{t-1})+1} \vvkt_s(\gv^{t-1}) + \frac{1}{\b_s(\gv^{t-1})+1} 
\frac{\uvctw_{\bar{1}s}(\gv^{t-1})}{\Psictw_{\bar{1}s}(\gv^{t-1})}
\frac{\uvctw_{1s}(\gv^{t-1})}{\Psictw_{1s}(\gv^{t-1})}.
\]
It is now enough to show that
\[
\frac{\uvctw_{s'}(\gv^{t-1})}{\Psictw_{s'}(\gv^{t-1})} =1
\text{ for }s'=\overline{\omega_{t-1-|s|}}s.
\]
This holds since $\uvctw_s=\Psictw_s$ for any off-path node $s\notin \rho(\omega_{t-D}^{t-1})$ by definition~\eqref{eq:uvctw_recursion}.\qed

\subsubsection{Proof of Proposition~\ref{prop:ctw_potential_ratio}}
Similar to the processing betas algorithm~\citep{Willems--Tjalkens--Ignatenko2006}, we only need to show that 
\[
\frac{\Psictw_{\overline{\omega_{t-1-|s|}}s}(\gv^{t})}{\Psictw_{\overline{\omega_{t-1-|s|}}s}(\gv^{t-1})} =1
\text{ for }s'=\overline{\omega_{t-1-|s|}}s
\text{ for any $s\notin \rho(\omega_{t-D}^{t-1})$}.
\]
Since the new symbol $\gv_t$ is added to a node $s$ if and only if $s\in\rho(\omega_{t-D}^{t-1})$, if $s\notin \rho(\omega_{t-D}^{t-1})$, then the CTW potential on the node $s$ will not be updated. This proves the claim.
\qed

\subsection{Technical lemmas}

\begin{lemma}[{\citealp[Lemma~10]{Orabona--Pal2016}}]
\label{lem:extremes}
Let $h\suchthat (-a,a)\to\Real$ be an even, twice differentiable function that satisfies $xh''(x)\ge h'(x)$ for all $x\in[0,a)$.
Let $c\suchthat [0,\infty)\times[0,\infty)\to\Real$ be an arbitrary function.
If $u,v\in\Hc$ satisfy $\|u\|+\|v\|<a$, then
\begin{align*}
c(\|u\|,\|v\|)\cdot\langle u,v\rangle - h(\|u+v\|)
&\ge \min_{r\in\{\pm 1\}} \{
rc(\|u\|,\|v\|)\|u\|\|v\| - h(\|u\|+r\|v\|)
\}.
\end{align*}
\end{lemma}

\begin{proof}[Proof sketch]
It is easy to check that the inequality holds if $u=0$ or $v=0$.
Hence, we assume $u,v\neq 0$. With 
$\a \defeq \langle u,v\rangle/(\|u\|\|v\|)$, we can write the left hand side of the desired inequality as 
\[
f(\a) \defeq c(\|u\|,\|v\|)\|u\|\|v\| \a - h(\sqrt{\|u\|^2+\|v\|^2+2\a\|u\|\|v\|}).
\]
Since the function $h$ is assumed to be even, it is equivalent to showing that
\[
\inf_{\a\in[-1,1]} f(\a) = \min\{f(+1),f(-1)\}.
\]
By using the condition $xh''(x)\ge h'(x)$, one can easily show that $f$ is concave by checking $f''(\a)\le 0$, which concludes the proof.
\end{proof}

\begin{lemma}[{\citealp[Example 13.7]{Bauschke--Combettes2011}}]
\label{lem:dual_composition_dual}
Let $\phi\suchthat\Real \to(-\infty,+\infty]$ be even. Then $(\phi\circ\|\cdot\|)^\star=\phi^\star\circ\|\cdot\|$.
\end{lemma}

\begin{lemma}[{\citealp[Lemma 5.8]{Orabona2019}}]
\label{lem:dual_scale}
Let $f$ be a function and let $f^\star$ be its Fenchel conjugate.
For $a>0$ and $b\in\Real$, the Fenchel conjugate of $g(x)=af(x)+b$ is $g^\star(z)=af^\star(z/a)-b$.
\end{lemma}

\begin{lemma}[{\citealp[Theorem~5.8]{Orabona2019}}]
For a convex, proper, closed function $h\suchthat \Real^d\to(-\infty,+\infty]$, we have $\langle \th,x\rangle \ge h(x)+h^\star(\th)$, where the equality is attained if and only if $x\in \partial h^\star(\th)$. 
\end{lemma}


Since $f(x)\ge h(x)$ for any $x\in\Real$ implies $f^\star(u)\ge h^\star(u)$ for any $u\in\Real$, it is enough to find the conjugate dual of a function $h(x)=\b\exp(\frac{x^2}{2\a})$ for $\a,\b>0$.

The \emph{Lambert function} $W\suchthat (-1/e,\infty)\to[0,\infty)$ is defined by the equation $x=W(x)e^{W(x)}$ for $x\ge 0$.

\begin{lemma}[{\citealp[Lemma 17]{Orabona--Pal2016}}]
For $x\ge 0$,
\[
0.6321 \ln(x+1) \le W(x) \le \ln(x+1).
\]
\end{lemma}

\begin{remark}
Here, $0.6321\ldots\approx 1/b^*$, where $b^*$ is the solution to the equation
\[
\frac{eb}{(e+1)b+1} = \frac{b}{(b+1)\ln(b+1)}.
\]
\end{remark}

\begin{proposition}[{\citealp[Lemma 18]{Orabona--Pal2016}}]
\label{prop:conjugate_base_function}
For $h(x)=\b\exp(\frac{x^2}{2\a})$ with $\a,\b>0$,
\[
h^\star(y)
=y\sqrt{\a W\Bigl(\frac{\a y^2}{\b^2}\Bigr)} - \b\exp\Bigl(\half W\Bigl(\frac{\a y^2}{\b^2}\Bigr)\Bigr)
=y\sqrt{\a}\Bigl(\sqrt{W\Bigl(\frac{\a y^2}{\b^2}\Bigr)} - \sqrt{\frac{1}{W(\frac{\a y^2}{\b^2})}}\Bigr).
\]
In particular, 
\[
h^\star(y) \le y\sqrt{\a \ln\Bigl(\frac{\a y^2}{\b^2}+1\Bigr)} - \b.
\]
\end{proposition}

For a generalization with the product potential, we also have the following proposition.
\begin{proposition}
\label{prop:conjugate_base_function_product}
Define $f_i(y_i)=\b_i\exp(\frac{y_i^2}{2\a_i})$ with $\a_i,\b_i>0$ for each $i\in S$, and define
$f(y_1,\ldots,y_S)=f_1(y_1)\cdots f_S(y_S)$.
Then, we have
\[
f^\star(y_1,\ldots,y_S)
=\sqrt{\a_1y_1^2+\ldots+\a_Sy_S^2}\Bigl(
\sqrt{W\Bigl(\frac{\a_1y_1^2+\ldots+\a_Sy_S^2}{\b_1^2\cdots\b_S^2}\Bigr)}
-\frac{1}{\sqrt{W\Bigl(\frac{\a_1y_1^2+\ldots+\a_Sy_S^2}{\b_1^2\cdots\b_S^2}\Bigr)}}
\Bigr).
\]
In particular, 
\[
f^\star(y_1,\ldots,y_S)
\le \sqrt{(\a_1y_1^2+\ldots+\a_Sy_S^2)\ln\Bigl(\frac{\a_1y_1^2+\ldots+\a_Sy_S^2}{\b_1^2\cdots\b_S^2}+1\Bigr)} - \b_1\cdots\b_S
\]
\end{proposition}

\begin{proof}
For the sake of simplicity, we prove only for $S=2$. 
The proof can be generalized to any $S\ge 2$ with little modification.
To find
\[
f^\star(y_1,y_2)=\sup_{x_1,x_2} (y_1x_1+y_2x_2 - f_1(x_1)f_2(x_2)),
\]
we consider the stationarity conditions
\[
\frac{\partial}{\partial x_i} (y_1x_1+y_2x_2 - f_1(x_1)f_2(x_2)) = 0
\]
for $i\in \{1,2\}$, which leads to 
\[
\begin{cases}
y_1&=f_1'(x_1)f_2(x_2),\\
y_2&=f_1(x_1)f_2'(x_2).
\end{cases}
\]
Since $f_i'(x)=\frac{x}{\a_i}f_i(x)$, we have
\[
\begin{cases}
y_1&=\frac{x_1}{\a_1}f_1(x_1)f_2(x_2),\\
y_2&=\frac{x_2}{\a_2}f_1(x_1)f_2(x_2).
\end{cases}
\]
Manipulating the equations, we have
\[
\Bigl(\frac{x_1^2}{\a_1}+\frac{x_2^2}{\a_2}\Bigr)\exp\Bigl(\frac{x_1^2}{\a_1}+\frac{x_2^2}{\a_2}\Bigr)=\frac{\a_1 y_1^2+\a_2 y_2^2}{\b_1^2\b_2^2},
\]
which leads to
\[
\frac{x_1^2}{\a_1}+\frac{x_2^2}{\a_2} = W\Bigl(\frac{\a_1y_1^2+\a_2y_2^2}{\b_1^2\b_2^2}\Bigr).
\]
Hence, 
\[
f(x_1^*,x_2^*)=\b_1\b_2\exp\Bigl(\half W\Bigl(\frac{\a_1y_1^2+\a_2y_2^2}{\b_1^2\b_2^2}\Bigr)\Bigr)
=\sqrt{\frac{\a_1y_1^2+\a_2y_2^2}{W(\frac{\a_1y_1^2+\a_2y_2^2}{\b_1^2\b_2^2})}}.
\]
Finally, we can compute 
\[
y_1x_1^*+y_2x_2^*
=\frac{\a_1 y_1^2+\a_2y_2^2}{f(x_1^*,x_2^*)}
=\sqrt{(\a_1y_1^2+\a_2y_2^2)W\Bigl(\frac{\a_1y_1^2+\a_2y_2^2}{\b_1^2\b_2^2}\Bigr)},
\]
whence
\begin{align*}
f^\star(y_1,y_2)
&=y_1x_1^*+y_2x_2^*-f(x_1^*,x_2^*)
\\&
=\sqrt{\a_1y_1^2+\a_2y_2^2}\Bigl(
\sqrt{W\Bigl(\frac{\a_1y_1^2+\a_2y_2^2}{\b_1^2\b_2^2}\Bigr)}-\frac{1}{\sqrt{W(\frac{\a_1y_1^2+\a_2y_2^2}{\b_1^2\b_2^2})}}
\Bigr).
\qedhere
\end{align*}
\end{proof}

\begin{lemma}[{\citealp[Lemma~9.4]{Orabona2019}}]
\label{lem:best_wealth_upper_bound}
For any $T\ge 1$ and any $c^T\in[-1,1]^T$, we have
\[
\max_{b\in[-1,1]} \prod_{t\in[T]} (1+bc_t)
\le \exp\Bigl(\frac{\ln 2}{T}(\nsum c^T)^2\Bigr).
\]
\end{lemma}

\newpage
\section{THE CTW OLO ALGORITHM}

\def\NoNumber#1{{\def\alglinenumber##1{}\State #1}\addtocounter{ALG@line}{-1}}

\begin{algorithm}[htb]
    \caption{CTW OLO algorithm}\label{alg:ctw_olo}

    \Hyperparameter maximum depth $D\ge 1$, auxiliary sequence $\Omega=(\omega_t)_{t\ge 1}$, initial wealth $\wealth_0>0$.
    
    
    \begin{algorithmic}[1] 
        \Procedure{CtwOlo}{$D,\Omega,\wealth_0$}
            \State Initialize a context tree $\Tc_D$ of depth $D$ with $G_s\gets \phi$ and $\b_s\gets 1$ for each $s\in\Tc_D$
            \ForEach{$t=1,2,\ldots$}
                \State Compute $\vvctw(\gv^{t-1})=\vvctw_\lambda(\gv^{t-1})$ by computing,
                for $s_0,\ldots,s_D\in \rho(\omega_{t-D}^{t-1})$,
                \[
                \vvctw_{s_d}(\gv^{t-1})\gets
                \begin{cases}
                \frac{\b_{s_d}(\gv^{t-1})}{\b_{s_d}(\gv^{t-1})+1} \vvkt_{s_d}(\gv^{t-1}) + \frac{1}{\b_{s_d}(\gv^{t-1})+1} \vvctw_{s_{d+1}}(\gv^{t-1}) & \text{if }d<D\\
                \vvkt_{s_D}(\gv^{t-1}) & \text{if }d=D
                \end{cases}
                \tag{\ref{eq:vvctw_recursion}}
                \]

                \State Set $\xvctw_t(\gv^{t-1})\gets \vvctw(\gv^{t-1})\wealth_{t-1}$
                \State Receive $\gv_t$ and update the cumulative wealth 
                $\wealth_t\gets \wealth_{t-1}+\langle \gv_t,\xvctw_t(\gv^{t-1})\rangle$
                \State Update $G_s\gets G_s+\gv_t$ and update $\b_s$ for $s_d=\omega_{t-d}^{t-1}$, $d=0,\ldots,D-1$, as
                \[
                \b_{s_d}(\gv^{t-1})\gets
                \b_{s_d}(\gv^t)
                =\b_{s_d}(\gv^{t-1}) \frac{\Psikt_{s_d}(\gv^t)}{\Psikt_{s_d}(\gv^{t-1})} \frac{\Psictw_{s_{d+1}}(\gv^{t-1})}{\Psictw_{s_{d+1}}(\gv^{t})},
                \tag{\ref{eq:beta_update_recursion}}
                \]
                
                \qquad where
                \[
                \frac{\Psictw_{s_d}(\gv^{t})}{\Psictw_{s_d}(\gv^{t-1})}
                =\begin{cases}
                \frac{\b_{s_d}(\gv^{t-1})}{\b_{s_d}(\gv^{t-1})+1} \frac{\Psikt_{s_d}(\gv^t)}{\Psikt_{s_d}(\gv^{t-1})}
                +\frac{1}{\b_{s_d}(\gv^{t-1})+1}\frac{\Psictw_{s_{d+1}}(\gv^{t})}{\Psictw_{s_{d+1}}(\gv^{t-1})} & \text{if }d<D\\
                \frac{\Psikt_{s_D}(\gv^t)}{\Psikt_{s_D}(\gv^{t-1})} & \text{if }d=D
                \end{cases}
                \tag{\ref{eq:ctw_potential_ratio_recursion}}
                \]
                
                \qquad for $s_d=\omega_{t-d}^{t-1}$, $d=0,\ldots,D$
                \State Receive $\omega_t$
            \EndFor
        \EndProcedure
    \end{algorithmic}
\end{algorithm}

\newpage

\section{EXPERIMENT DETAILS AND ADDITIONAL FIGURES}
\label{app:sec:exp}
\paragraph{Problem setting} 
We applied the proposed OLO algorithms to solve the online linear regression problem as described in Appendix~\ref{sec:other_per_state_algos} especially with absolute loss $\ell_t(\wv_t)=|\langle \wv_t,\xv_t\rangle -y_t|$, where $\wv_t$ denotes the action of an OLO algorithm and $\xv_t$ denotes the feature vector.
Hence, we linearized the convex loss and fed the subgradient $\partial\ell_t(\wv_t)=\sgn(\langle \wv_t,\xv_t\rangle -y_t)\xv_t$ to an OLO algorithm.

\paragraph{Data preprocessing} 
For each dataset, we linearly interpolated any missing values. We discarded time stamps as well as some categorical features such as \texttt{cbwd} of Beijing PM2.5 and \texttt{weather\_description} of Metro Inter State Traffic Volume, and binarized the others, if possible, such as \texttt{holiday}, \texttt{weather\_main}, and \texttt{snow\_1h} of Metro Inter State Traffic Volume.
We also applied a logarithmic mapping $x\mapsto\ln(1+x)$ for the features \texttt{lws}, \texttt{ls}, \texttt{lr} of Beijing PM2.5 and applied another logarithmic mapping $x\mapsto \ln x$ to the feature \texttt{rain\_1h}, to make the features more suitable for linear regression.
We then normalized each feature $\tilde{\xv}_t$ so that $\|\tilde{\xv}_t\|_2=1$ and added all-one coordinates as the bias component with an additional scaling by $1/\sqrt{2}$. 
After this preprocessing step, we obtained 7-dimensional feature vectors for both datsets.
See the attached Python code for the details in Supplementary Material.

\paragraph{Computing resource}
All experiments were run on a single laptop with a CPU {Intel(R) Core(TM) i7-9750H CPU \@ 2.60GHz} with 12 (logical) cores and 16GB of RAM.

\begin{figure}[ht]
    \flushleft
    ~~\includegraphics[width=.8\textwidth]{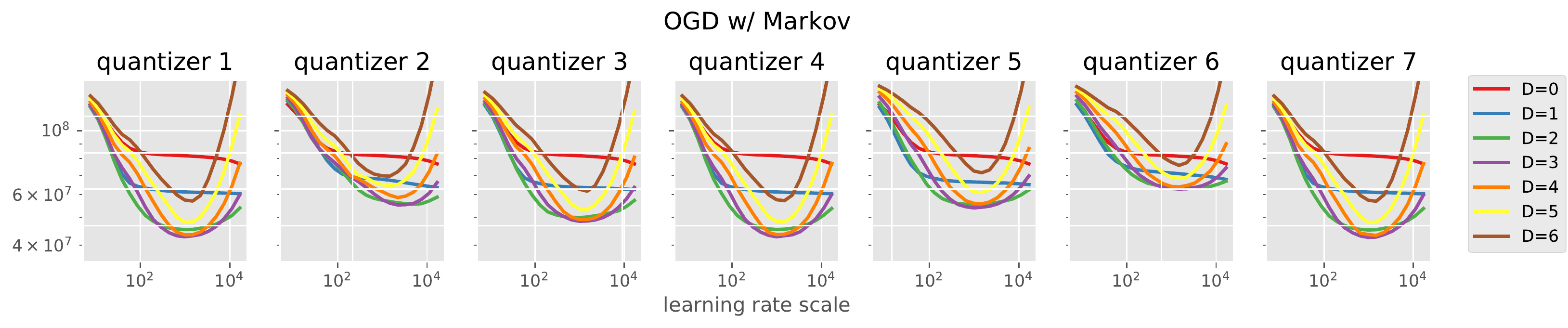}\\
    \includegraphics[width=.9\textwidth]{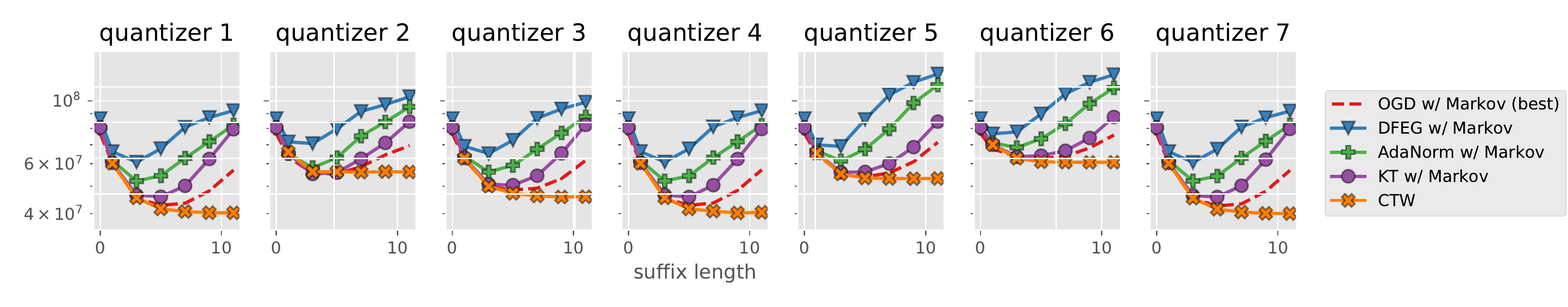}
    \vspace{-1.5em}
    \caption{Metro Inter State Traffic Volume dataset~\citep{Hogue2019}. 
    The $y$-axes represent cumulative losses.
    (a) Performance of per-state OGD adaptive to Markov side information with various learning rate scales. 
    (b) Performance of parameter-free algorithms.}
    \label{fig:add_exps_metro}
\end{figure}

\begin{figure}[ht]
    \flushleft
    ~~\includegraphics[width=.8\textwidth]{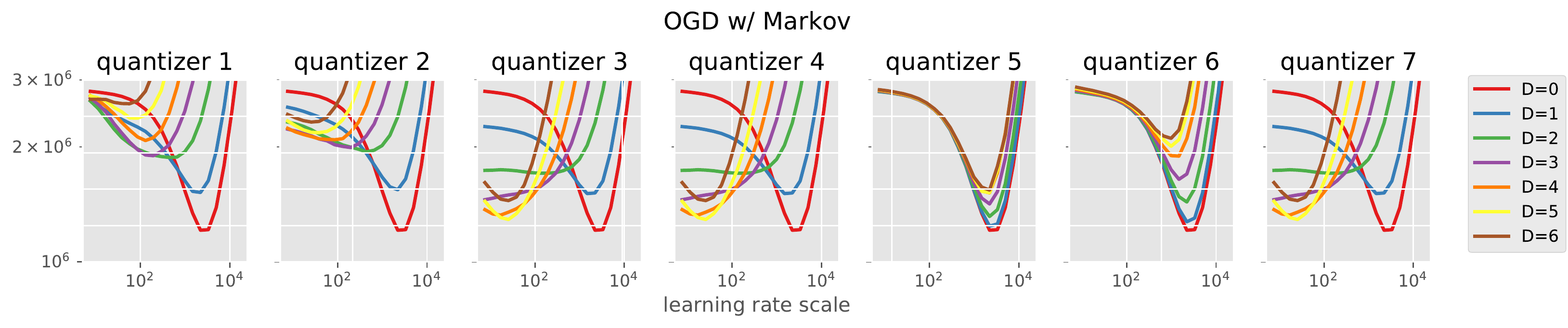}\\
    \includegraphics[width=.9\textwidth]{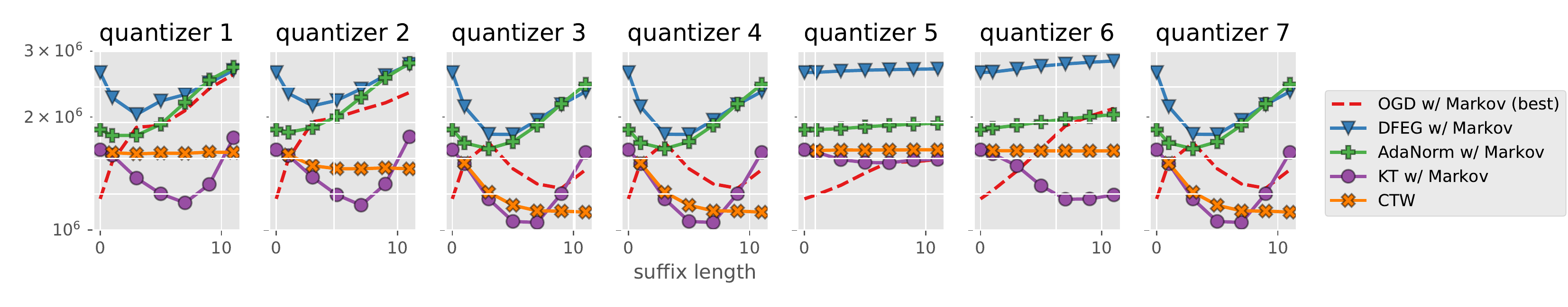}
    \vspace{-1.5em}
    \caption{Beijing PM2.5 dataset~\citep{Liang--Zou--Guo--Li--Zhang--Zhang--Huang--Chen2015}. 
    See the caption of Figure~\ref{fig:add_exps_metro} for details.}
    \label{fig:add_exps_beijing}
\end{figure}

\fi

\end{document}